\documentclass[10pt,transaction,twocolumn]{IEEEtran}
\usepackage{amsmath,amsfonts,amssymb,mathrsfs,amsthm}
\usepackage{latexsym}
\usepackage[nospace,noadjust]{cite}
\usepackage{latexsym}
\usepackage{epsfig}
\usepackage{hyperref}
\usepackage{epstopdf}
\usepackage{environ}
\usepackage[nospace,noadjust]{cite}
\usepackage{array}
\usepackage[english]{babel}
\usepackage{color}
\usepackage{babel}
\usepackage{graphicx,xcolor,colortbl}    
\usepackage{epstopdf,epsfig}
\usepackage{graphicx}
\usepackage{tabularx,ragged2e,booktabs,caption}
\usepackage{algorithm,algpseudocode}
\usepackage{tikz}
\usetikzlibrary{calc}
\usepackage{graphicx}
\usepackage{caption}
\usepackage{hhline}
\usepackage{caption}
\usepackage{subcaption}
\usepackage{pgfplots} 
\usepackage{tikz}
\usetikzlibrary{calc}
\pgfplotsset{compat=newest} 
\pgfplotsset{plot coordinates/math parser=false} 
\newlength\figureheight 
\newlength\figurewidth 
\usetikzlibrary{shapes,arrows}
\usepackage{xcolor,colortbl}
\newtheorem{theorem}{{\bf Theorem}}
\newtheorem{assumption}{{\bf Assumption}}

\newtheorem{proposition}{\noindent {\bf Proposition}}
\newtheorem{lemma}{\noindent {\bf Lemma}}

\include{new_commands}
\newcommand{\diag}{\mathop{\rm diag}}
\newcommand{\tr}{\mathop{\rm tr}}

\makeatletter
\let\l@ENGLISH\l@english
\makeatother

\begin{document}%
\title{Blind Measurement Selection: A Random Matrix Theory Approach}
\author{Khalil Elkhalil,~\IEEEmembership{Student Member,~IEEE,} Abla Kammoun,~\IEEEmembership{Member,~IEEE}, Tareq~Y.~Al-Naffouri,~\IEEEmembership{Member,~IEEE}, and Mohamed-Slim Alouini,~\IEEEmembership{Fellow,~IEEE}

\thanks{
K. Elkhalil, A. Kammoun, T. Y. Al-Naffouri and M.-S. Alouini are with the Electrical Engineering Program, King Abdullah University of Science and Technology, Thuwal, Saudi Arabia; e-mails: \{khalil.elkhalil, abla.kammoun, tareq.alnaffouri, slim.alouini\}@kaust.edu.sa.
}

}

\maketitle
\vspace{-15mm}
\begin{abstract}
    This paper considers the problem of selecting a set of $k$ measurements from $n$ available sensor observations. The selected measurements should minimize a certain error function assessing the error in estimating a certain $m$ dimensional parameter vector.   The exhaustive search inspecting each of the $n\choose k$ possible choices would require a very high computational complexity and as such is not practical for large $n$ and $k$. Alternative methods with low complexity
    have recently been investigated but their main drawbacks are that 1) they require perfect knowledge of the measurement matrix  and 2) they need to be
applied at the pace of change of the measurement matrix. To overcome these issues, we consider the asymptotic regime in which $k$, $n$ and $m$ grow large at the same pace. Tools from random matrix theory are then used to approximate in closed-form the  most important error measures that are commonly used. The asymptotic approximations are then leveraged  to select properly $k$ measurements exhibiting low values for the asymptotic error measures.  Two heuristic  algorithms are  proposed: the first one
merely consists in applying 
the convex optimization  artifice to the asymptotic error measure. The second algorithm is a low-complexity greedy algorithm that attempts to look for a sufficiently good solution for the original minimization problem. The greedy algorithm can be applied to both the exact and the asymptotic error measures and can be thus implemented in  blind and channel-aware fashions.       
We present two potential applications where the proposed algorithms can be used, namely antenna selection for uplink transmissions in large scale multi-user systems and sensor selection for wireless sensor networks. Numerical results are also presented and sustain the efficiency of the proposed blind methods in reaching the performances of channel-aware algorithms. 
\end{abstract}
\section{Introduction} \label{introduction}
\PARstart{M}{measurement} selection is an old concept that finds its roots in many applications such as robotics, wireless sensor networks and wireless communications to name a few \cite{boyd,robotics_ref,kammer}. It aims to reduce the complexity of  the estimation problem in linear models where the $n$-dimensional response vector is linearly related to the unknown $m$-dimensional vector. 
 The reduction in computational complexity is achieved by using
only the $k$ measurements that minimize a certain given error function assessing the quality of the selected measurements. One naive approach for solving the measurement selection problem is to go through all $n\choose k$ possible selections and select the ones that present the lowest achievable value of a given  error measure. This procedure, though being optimal, is not practical especially when high dimensional observations are considered. In general, it seems that looking for an optimal solution is expected to call
for solving an
NP-hard problem as asserted by \cite{boyd}, which gives little hope of determining the optimal solution using a polynomial complexity algorithm. 
 As a result, attention has turned to sub-optimal alternatives to solve the measurement selection problem. In this vein, genetic algorithms  using some local search methods have been proposed in  \cite{genetic}. Local optimization techniques have been also proposed in   \cite{nguyen,john}. 
 Although these methods can have good performance with modest complexities, they do not guarantee any theoretically achievable bound on the performance. The first achievable bound has been
derived in \cite{boyd} where the authors in \cite{boyd} resorted to convex relaxation artifice. This has led to a convex problem that can be solved with a complexity growing as $\mathcal{O}\left(n^3\right)$.  

The aforementioned algorithms allow good performances coupled with a lower complexity as compared to exhaustive search. However, in order to accurately evaluate the overall complexity, it is important to consider how often the selection procedure should be repeated. For fast-fading varying linear models, the underlying measurement matrix, capturing the linear dependence between the input and output vectors, changes at a rapid pace. Hence, the overall complexity should be scaled by the
number of times the measurement matrix changes over a given time window, which can result in a prohibitively high computational complexity.
To overcome these issues, we propose in this work blind selection algorithms that leverage the statistics of the measurement matrix rather than its instantaneous realization. This can be for instance useful when for some practical concerns, it is not possible to acquire the measurement matrix.  The main idea behind the proposed blind methods lies in the observation that the considered error measure depending on the random measurement matrix can be approximated by some deterministic quantities
depending solely on the statistics of the measurement matrix. This fact is supported by  results from random matrix theory, confirming the accuracy of the approximation in large dimensional settings.  Using this theory, we show that the most used error measures can be approximated in closed-form by deterministic quantities depending soleley on the correlation between the columns of the measurement matrix. Interestingly, it turns out that, as far as the asymptotic regime is
con, it concerned the correlation matrix that
holds all the information about the best set of measurements to be selected. Particularly, it is shown that if the columns of the measurement matrix are uncorrelated, any randomly selected subset of $k$ measurements would asymptotically exhibit the same performances. It thus unfolds that the benefit from optimizing over the set of measurements to be selected is more significant in case of high correlation between the columns of the measurement matrix.          

 \par
 Based on the obtained asymptotic approximations of the considered error measures, we propose two different blind approaches. The first one is mostly inspired from the work of \cite{boyd} and consists in applying the convex artifice to the asymptotic error measures. However, the optimization of the resulting problem might  not be tractable, as there is no guarantee of the convexity of the asymptotic error measure. We, therefore, establish in this paper the convexity of the asymptotic error
 measure, which opens up the possibility of using standard convex optimization tools. The second algorithm is a greedy algorithm that attempts to get close to the optimal solution within a few iterations. Interestingly, the greedy algorithm can be also applied to the exact error measures, and can be thus implemented in both blind and channel-aware scenarios. It is shown that not only the greedy algorithm presents lower complexity but it also achieves higher performances than the
 convex-relaxation based algorithm when implemented in either blind or channel-aware modes. 
 \par 
 The proposed algorithms can be used in many applications. We select in this paper two potential applications where the measurement selection problem arise. The first one concerns the design of low complexity linear receivers for uplink large-scale multi-user MIMO systems, better known as Massive MIMO systems. Such systems are gaining an increasing interest and constitute promising candidates for future wireless systems, primarily due to their abilities of achieving remarkable
 performance enhancements in terms of capacity, radiated energy efficiency and link reliability \cite{larsson,marzetta,rusek}. However, while the use of multiple antennas  allows to significantly improve the spatial diversity, it comes inevitably at the cost of a higher computational complexity, which might call into question the feasibility of such systems \cite{gao}. Besides, it is not even clear whether the performance enhancement  is worthy  of using all antenna resources. It might happen
 in many scenarios that some antennas undergo severe fading, and as such, selecting a subset of antennas will result in substantial saving in complexity without sacrificing performance. Selecting these antennas is the crucial achievement of antenna selection algorithms. The use of these algorithms  has already been  advocated in \cite{molisch,sanayei} as an efficient solution  to reduce the number of RF chains in conventional MIMO systems, leading to a significant reduction in complexity and costs while preserving most of the potential of full MIMO systems.  The second application concerns the problem of sensor selection in wireless sensor networks (WSN), already
 extensively studied in the literature of signal processing. A question of interest  in this field, is how to find the optimal placement of $k$ sensors, assuming  $n$ available sensors in different locations \cite{yao}. 

 To sum up, the major  contributions of this paper are summarized as follows:
\begin{enumerate}
    \item We provide accurate asymptotic approximations for the three most used error measures, namely the Mean Square Error (MSE), the Log Volume of the Confidence Ellipsoid (LCE) and the Worst Case Error Variance (WEV). These approximations depend solely on the statistics of the measurement matrix and are shown to be convex.  
    \item Based on the provided approximations, we propose two blind algorithms to perform antenna selection without knowledge of the instantaneous measurement matrix. The first one is based on the concept of convex relaxation while the second one is a greedy iterative algorithm that attempts to get close to the optimal solution.  
	\item We study the complexity of both algorithms and show that the greedy algorithm achieves quadratic complexity. 
    \item We select two applications in which the problem of sensor selection arise, namely antenna selection in massive MIMO systems and sensor selection in WSN. We show how the proposed  algorithms can be used in both applications and compare their performances with channel aware algorithms which assume perfect knowledge of the measurement matrix. 
        We show that as long as the correlation between the columns of the measurement matrix is high, the average performance is close to that of channel
        aware algorithms.  
\end{enumerate}
The remainder of the paper is organized as follows. In section \ref{formulation}, we state the measurement selection problem and present some related works. In section \ref{blind_measurement}, we provide asymptotic approximations for the asymptotic error measures, based on which we propose two different blind algorithms to perform measurement selection. 
Finally, and prior to concluding the paper in section \ref{conclusion}, we discuss in section \ref{application} two potential applications for the proposed blind approach.
\par
\emph{Notations:} Throughout the paper, we use the following notations: Vectors are denoted by lower case bold letters and matrices are denoted by bold capital letters ($\mathbf{I}_n$ is the identity matrix of size $n$). For a given matrix $\mathbf{A}$, we refer by $\left[\mathbf{A}\right]_{i,j}$ its $\left(i,j\right)$th entry, and use $\mathbf{A}^T$ and $\mathbf{A}^H$ to denote its transpose and Hermitian respectively. We respectively denote by $\left \|. \right \|$, $\det \left(.\right)$ and $\tr\left(.\right)$, the spectral norm, the determinant and the trace of a matrix. Finally, we denote by $\diag\left(\mathbf{a}\right)$, the diagonal matrix with diagonal elements, the entries of $\mathbf{a}$.

\section{Problem formulation and related works} \label{formulation}
We consider the problem of estimating  an $m-$dimensional vector $\mathbf{x} \in \mathbb{C}^{m\times 1}$ from  measurements when the observed vector $\mathbf{y}\in \mathbb{C}^{n \times 1}$ is  related to $\mathbf{x}$ by the following relation:
\begin{equation}
\label{linear1}
\mathbf{y}=\mathbf{H}\mathbf{x}+\mathbf{v},
\end{equation}
Herein $\mathbf{H}=\left \{ h_{i,j} \right \} \in \mathbb{C}^{n \times m}$ denotes the measurement matrix and is assumed to be Gaussian with one-side correlation $\mathbf{R}$, i.e., $\mathbf{H}=\mathbf{R}^{\frac{1}{2}} \mathbf{W}$, where $\mathbf{W} \in \mathbb{C}^{n \times m}$ is a matrix with i.i.d zero mean unit variance entries and $\mathbf{v}$ is the additive noise vector with independent, zero mean, unit variance,  circularly symmetric complex Gaussian entries, i.e. $\mathbf{v} \sim \mathcal{CN}\left(\mathbf{0}_{n \times 1},\mathbf{I}_n\right)$.
Assuming that the number of measurements $n$ exceeds the signal dimension $m$, i.e., $n > m$, the estimate of $\mathbf{x}$ denoted by $\widehat{\mathbf{x}}$ can be recovered by using the least square (LS) estimator as \cite{sayed}
\begin{equation}
\label{blue}
\begin{split}
\widehat{\mathbf{x}}& = \left(\mathbf{H}^H\mathbf{H}\right)^{-1}\mathbf{H}^{H} \mathbf{y}  \\
& =\left(\mathbf{W}^H\mathbf{R}\mathbf{W}\right)^{-1}\mathbf{W}^H \mathbf{R}^{\frac{1}{2}} \mathbf{y} \\
& = \mathbf{x} + \left(\mathbf{W}^H\mathbf{R}\mathbf{W}\right)^{-1}\mathbf{W}^H \mathbf{R}^{\frac{1}{2}} \mathbf{v}.
\end{split}
\end{equation}
Clearly, the estimation error, $\mathbf{x-\widehat{x}}$ is zero mean with  covariance:
\begin{equation}
\label{covariance}
\mathbf{\Sigma} = \left(\mathbf{W}^H\mathbf{R}\mathbf{W}\right)^{-1}.
\end{equation}
\subsection{Measurement selection}
Measurement selection intends to select the $k$ \emph{best} measurements that are the most representative in the sense that they constitute the set of $k$ measurements minimizing a certain given error.  Once selected, these measurements will be used in place of the whole available data vector. By reducing the available number of measurements, measurement selection inevitably induces  a performance loss, but allows a reduction in the computational complexity. This
becomes all the more important in several applications such as large-scale MIMO systems, where the treatment of the whole number of measurements received by the large antenna array might be not possible. 
To mathematically formulate measurement selection, we define the selection matrix $\mathbf{S} \in \mathbb{R}^{k \times n}$ as the matrix that permits to extract $k$ measurements from $\mathbf{y}$. Let $\mathcal{S}$ denote the set of indexes of cardinality $k$ containing the indexes of measurements to be selected. The selected measurement vector is thus  given by 
\begin{equation}
\label{selected}
\mathbf{y}_{\mathcal{S}} = \mathbf{S}\mathbf{y},
\end{equation}
where $\mathbf{S}$ is defined as follows
\begin{equation}
\label{select_matrix}
\left[\mathbf{S}\right]_{i,j}=\left\{\begin{matrix}
1 & j=\mathcal{S}\left[i \right ] \\ 
0 &  \textnormal{otherwise} 
\end{matrix}\right., i=1,\cdots,k.
\end{equation}
where $\mathcal{S}[i]$ denotes the $i$-th element in set $\mathcal{S}$.
Based on the structure of $\mathbf{S}$ in (\ref{select_matrix}), we have the following properties
\begin{itemize}
	\item $\mathbf{S}\mathbf{S}^T=\mathbf{I}_k$. 
	\item $\mathbf{S}^T\mathbf{S}=\diag\left(\mathbf{s}\right)$.
\end{itemize}
where $\mathbf{s}=\left \{ s_i \right \}_{i=1,\cdots,n}$ is a $n-$dimensional vector with entries equal to $1$ at the locations given by $\mathcal{S}$ and zeros elsewhere.
The LS estimator $\widehat{\bf x}_{\mathcal{S}}$ obtained from using the selected measurement vector ${\bf y}_{\mathcal{S}}$ is :
$$
\widehat{\bf x}_{\mathcal{S}}=\left({\bf H}^{H}{\rm diag}({\bf s}){\bf H}\right)^{-1}{\bf H}^{H}{\bf S}^{T}{\bf y}_S
$$
Upon applying the operator defined by $\mathbf{S}$, the resulting error covariance matrix, which we denote by $\mathbf{\Sigma}_{\mathcal{S}}$ easily writes as 
\begin{equation}
\label{cov_selection}
\begin{split}
    \mathbf{\Sigma}({\bf s}) &=\left(\mathbf{H}^H{\rm diag}({\bf s }) \mathbf{H}\right)^{-1} \\
& = \left(\mathbf{W}^H \mathbf{R}^{\frac{1}{2}}\diag\left(\mathbf{s}\right)\mathbf{R}^{\frac{1}{2}}\mathbf{W}\right)^{-1}.
\end{split}
\end{equation}
It can be seen from (\ref{cov_selection}) that $\mathbf{\Sigma}_{\mathcal{S}}$ is a Gram matrix with one side correlation given by the matrix ${\bf R}^{\frac{1}{2}}\diag\left(\mathbf{s}\right)\mathbf{R}^{\frac{1}{2}}$. \\
Measurement selection consists in selecting the optimal set  $\mathcal{S}^*$ with cardinality $k$, or equivalently vector ${\bf s}$ with only $k$ non-zero elements equal to $1$, that minimizes a certain error measure assessing the estimation quality. In other words, the optimal set can be obtained as the solution of the following problem: 
\begin{equation}
\label{selection_pb}
\begin{split}
    &{\bf s}^*= \underset{{\bf s}\ \in \ \mathbb{R}^n}{\text{argmin}} \:\:\:\: f\left(\mathbf{\Sigma}({\bf s})\right)\\ 
    & \text{s.t.} \:\: \:\: \:\: \:\:  s_i\in\left\{0,1\right\}\\
    & \:\: \:\: \:\: \:\: \:\: \:\: \:  {\bf 1}^{T}{\bf s}=k
\end{split}
\end{equation}
where here $f$ denotes the considered measure function. 
There have been various measures proposed in the literature \cite{serfling} to assess the estimation quality of the LS. All of them heavily depend  on the covariance matrix of the error \label{covariance}. In this paper, we will focus on the following ones:
\begin{enumerate}
    \item The Mean Square Error (MSE):
        The mean square error is defined as the average euclidian distance between the estimated vector and ${\bf x}$. When only a set of $k$ measurements is employed, the MSE writes as: 
\begin{equation}
\label{MSE}
\begin{split}
    \text{MSE}({\bf s}) & = \tr \left({\bf W}^{H}{\bf R}^{\frac{1}{2}}{\rm diag}({\bf s}){\bf R}^{\frac{1}{2}}{\bf W}\right)^{-1}
 \end{split}
\end{equation}
\item The Log Volume of the confidence Ellipsoid (LCE) 
    For a Gaussian random vector $\widehat{\bf x}$ in $\mathbb{C}^m$ with mean ${\bf x}$ and covariance $\boldsymbol{\Sigma}$, the $\eta-$ confidence ellipsoid is a multi-dimendional generalization of the $\eta-$ confidence interval. It corresponds to the minimum volume ellipsoid that contains  $\widehat{\bf x}-{\bf x}$ with propability $\eta$ and is given by:
   
\begin{equation}
\label{ellipsoid}
\mathcal{E}=\left \{\mathbf{z}:\mathbf{z}^{H}\mathbf{\Sigma}^{-1}\mathbf{z} 
\leq \alpha
\right \},
\end{equation}
where $\alpha=F^{-1}_{\chi^2_{2m}}\left(\eta\right)$, $F_{\chi^2_{2m}}$ being the cumulative distribution function of a chi-squared random variable with $2m$ degrees of freedom. The volume of the $\eta-$confidence ellipsoid defined in (\ref{ellipsoid}) is  (see Serfling \cite{serfling})  
\begin{equation}
\label{volume}
\textbf{vol}\left(\mathcal{E}_{\alpha}\right)=\frac{\left(\alpha\pi\right)^{m}}{\Gamma\left(m+1\right)}\text{det} \left(\mathbf{\Sigma}\right),
\end{equation}
where $\Gamma\left(.\right)$ is the Gamma function. 
It appears from \ref{volume} that the  determinant of $\boldsymbol{\Sigma}$ plays the role played by the variance in one dimension hence its  name {\it generalized variance}. When vector $\widehat{\bf x}$ represents an estimate of a given parameter vector, the lowest is the generalized variance, the highest is the estimation quality. 
It might be more convenient in practice to work with the log of the volume of the $\eta-$ ellispoid. A good estimate is thus characterized by a small value of the log-epsilloid volume, which will be confused, from now on, with $\frac{1}{m}\log {\rm det}(\frac{1}{n}\boldsymbol{\Sigma})$. We define thus the log volume of the confidence ellipsoid associated with the vector of selected measurement $\widehat{\bf x}_{\mathcal{S}}$ as:
\begin{align*}
    {\rm LCE}({\bf s})&=-\frac{1}{m}\log\left({\rm det}\left(\frac{1}{n}{\bf H}^{H}{\rm diag}({\bf s}){\bf H}\right)\right)\\
             &=-\frac{1}{m}\log\left({\rm det}\left(\frac{1}{n}{\bf W}^{H}{\bf R}^{\frac{1}{2}}{\rm diag}({\bf s}){\bf R}^{\frac{1}{2}}{\bf W}\right)\right)
\end{align*}

\item {Worst Case Error Variance (WEV)}: The worst case error variance (WEV) quantifies the maximum  variance of error over all directions. It corresponds to the maximum eigenvalue of the error covariance matrix. The WEV associated with the vector of selected measurement $\widehat{\bf x}_{\mathcal{S}}$ is thus defined as \footnote{The normalization factor $\frac{1}{m}$ is considered herein to comply with the asymptotic growth regime of random matrix theory}:
    $$
    {\rm WEV}({\bf s})=\frac{1}{\lambda_{min}\left(\frac{1}{m}{\bf W}^{H}{\bf R}^{\frac{1}{2}}{\rm diag}({\bf s}){\bf R}^{\frac{1}{2}}{\bf W}\right)}
    $$
    \end{enumerate}
\subsection{Related Works}
The main literature related to the present paper is represented by the work in \cite{boyd}, 
The main idea in this work relies on the observation that it is the non-convex nature of the constraints, requiring the selection vector ${\bf s}$ to possess elements in $\left\{0,1\right\}$ that makes problem     (\ref{selection_pb}) intractable. To overcome this issue, \cite{boyd} solves instead a convex related problem obtained by substituting  the Boolean constraints $s_i \in \left \{ 0,1 \right \}$ by the convex constraints $0 \leq s_i \leq 1$: 
\begin{equation}
\label{boyd_relax}
\begin{split}
    &\widehat{\mathbf{s}}= \arg\min_{{\bf s} \ \in \ \mathbb{R}^{n}}f(\boldsymbol{\Sigma}({\bf s}))\\
& \text{s.t.} \:\: \:\: \:\: \:\: \:\: \:\: \:\:\mathbf{1}^T \mathbf{s}=k  \\
& \:\: \:\: \:\: \:\: \:\: \:\: \:\: \:\: \: \: \: 0 \leq s_i \leq 1, \ \  i=1,\cdots,n.
\end{split}
\end{equation}
It is worth mentioning that the output of the optimization in (\ref{boyd_relax}) yields a higher value than the maximum objective function in (\ref{selection_pb}), and as such can be viewed as a global upper bound on the performance. Moreover, the optimal vector $\widehat{\mathbf{s}}$ can contain real values not necessarily zeros and ones. In order to obtain the indexes of the selected measurements, one should order the entries of $\widehat{\mathbf{s}}$ and then assign ones to the $k$
greatest values and zeros to the remaining entries. This results in a feasible solution to  the selection problem in (\ref{selection_pb}) which yields a lower bound on the objective function. \\
To solve the problem in (\ref{boyd_relax}), one can resort to interior-point methods which require few tens of iterations to converge where each iteration is performed with a complexity of $\mathcal{O}\left(n^3\right)$ computations. For more details on convex relaxation, the readers are referred to \cite{boyd} and references therein. \par
\section{Blind measurement selection} \label{blind_measurement}
Previous works dealing with sensor selection have essentially been based on the assumption of perfect knowledge of the measurement matrix ${\bf H}$. This not only might  not be satisfied in practice but also can make the application of the previously proposed algorithms more difficult as they should be applied at every  change in the measurement matrix ${\bf H}$. 

To overcome this issue, we propose in this work blind methods that leverage the knowldege of the channel statistics to perform the selection. These methods rely heavily on advanced results of random matrix theory. 

{\it Main Idea}:
The idea behind the proposed methods hinges on the fact that as the dimensions of the channel matrix grow large, quantities depending on the measurement matrix become more predictable in that they can be well-approximated by deterministic quantities depending only on the channel statistics. 
Such deterministic quantities can be characterized by resorting to tools from random matrix theory.
In light of this observation, we propose in this work to compute, in closed form, accurate approximations of the three error measures, namely the MSE, LCE and WEV.  As we will see later, the asymptotic analysis can be leveraged to blindly select good sets of measurements. 

\subsection{Asymptotic analysis of the error measures}
In this section, we determine, in closed form, accurate approximations for the MSE, LCE and WEV. For technical purposes, we shall consider the following growth regime:
\begin{assumption}
    We assume both that $n$ and $m$ grow large while their ratio $\frac{n}{m}$ satisfies: 
$$
\frac{n}{m}\to c\in \left(1,\infty\right)
$$
We also assume that $k$ grows large with:
$$
0<\lim\inf \frac{k}{n}<\lim\sup \frac{k}{n}<1. 
$$
and 
$$
\lim\inf \frac{k}{m}> 1.
$$
\label{ass:regime}
\end{assumption}
The channel matrix ${\bf H}$ is assumed to follow the following statistical model:
\begin{assumption}
    \label{ass:model}
    ${\bf H}\in\mathbb{C}^{n\times m}$ is complex Gaussian matrix with one-side correlation ${\bf R}$, i.e,
    $$
    {\bf H}={\bf R}^{\frac{1}{2}}{\bf W}
    $$
    where ${\bf W}$ is a matrix with i.i.d. normally distributed entries having zero-mean and unit variance. Moreover, the matrix ${\bf R}$ satisfies the following conditions:
    \begin{enumerate}
        \item ${\bf R}$ has a bounded spectral norm, i.e,
            $$
           \sup_n  \|{\bf R}\| <\infty
            $$
        \item The normalized trace of ${\bf R}$ satisfies:
            $$
            \inf_n \frac{1}{n}{\rm tr}{\bf R} >0.
            $$
        \end{enumerate}
\end{assumption}
With these assumptions at hand, we are ready to analyze the asymptotic behavior for the MSE and LCE. 
\begin{lemma}\cite{silverstein}
     Let $\delta$ be the unique solution to the following equation
    \begin{equation}
    \delta=m\left(\tr\left[{\bf R}{\rm diag}(s)\left({\bf I}_n+\delta{\bf R}{\rm diag}({\bf s})\right)^{-1}\right]\right)^{-1}
    \label{eq:delta}
\end{equation}
Define $\overline{\rm MSE}({\bf s})$ as
$$
\overline{\rm MSE}({\bf s})=\delta
$$
    Then, under assumptions \ref{ass:regime} and \ref{ass:model}, ${\rm MSE}({\bf s})$ satisfies
    $$
    {\rm MSE}({\bf s}) - \overline{\rm MSE}({\bf s})\xrightarrow[n\to\infty]{a.s.} 0.
    $$
\end{lemma}
\begin{lemma}\cite[Proposition 4.2]{dumont-07}
    Let $\delta$ be defined as in \eqref{eq:delta}.
    Define $\overline{\rm LCE}({\bf s})$ as
    $$
\overline{\rm LCE}({\bf s})= -\frac{1}{m}\log\det\left({\bf I}_n+\delta{\bf R}{\rm diag}({\bf s})\right) +\log(c\delta)+1\
    $$
    Then, under assumptions \ref{ass:regime} and \ref{ass:model},
    $$
    {\rm LCE}({\bf s}) -\overline{\rm LCE}({\bf s}) \xrightarrow[n\to\infty]{a.s.} 0
    $$
 \end{lemma}
 To obtain an asymptotic equivalent for the WEV, the following technical assumption is additonally needed:
 \begin{assumption}
     Let $\lambda_{1}({\bf s}),\cdots,\lambda_{n}({\bf s})$ be the eigenvalues of ${\bf R}^{\frac{1}{2}}{\rm diag}({\bf s}){\bf R}^{\frac{1}{2}}$.
     We assume that the probability measure $\frac{1}{n}\sum_{i=1}^n \boldsymbol{\delta}_{\lambda_i}({\bf s})$ converge weakly to a probability measure $\mu_s$. Moreover, we assume all the eigenvalues of ${\bf R}^{\frac{1}{2}}{\rm diag}({\bf s}){\bf R}^{\frac{1}{2}}$ to be almost surely contain in the support of $\mu_s$ $Supp(\mu_s)$ or equivalently:
     $$
     \max_{i=1,\cdots,n} {\rm dist}(\lambda_i({\bf R}^{\frac{1}{2}}{\rm diag}({\bf s}){\bf R}^{\frac{1}{2}}),\mu_s))\xrightarrow[n\to\infty]{a.s.} 0
     $$
     
     \label{ass:eigenvalue}
 \end{assumption}
 With this assumption at hand, the {\rm WEV} can be approximated as
 \begin{lemma}\cite{BAI99}
     Under assumptions \ref{ass:regime}, \ref{ass:model} and \ref{ass:eigenvalue},
     $$
     {\rm WEV}({\bf s}) - \overline{\rm WEV}({\bf s}) \xrightarrow[n\to\infty]{a.s.} 0
     $$
     where $\overline{\rm WEV}({\bf s})$ is given by: 
     $$
     \overline{\rm WEV}({\bf s})=-\frac{1}{\eta} +\frac{1}{m} \tr \left[\mathbf{R}\diag\left(\mathbf{s}\right)\left(\mathbf{I}+\eta \mathbf{R}\diag\left(\mathbf{s}\right)\right)^{-1}\right]
     $$
     and $\eta$ is the solution to the following equation in $\left(0,\infty\right)$
     $$
     \eta^2=\left(\frac{1}{m}\tr\left[\left({\bf R}^{\frac{1}{2}}{\rm diag}({\bf s}){\bf R}^{\frac{1}{2}}\right)^2\left({\bf I}_n+\eta{\bf R}^{\frac{1}{2}}{\rm diag}({\bf s}){\bf R}^{\frac{1}{2}}\right)^{-2}\right]\right)^{-1}.
     $$
 \end{lemma}
 \subsection{Blind selection techniques}
The asymptotic analysis carried out in the previous section is now leveraged to build efficient blind methods for measurement selection. Our blind approaches are based on solving the following selection problem:
\begin{equation}
    \label{eq:blind}
\begin{split}
    &{\bf s}^*= \underset{\mathbf{s}}{\text{argmin}} \overline{f}({\bf s}) \\ 
& \text{s.t.} \:\: \:\: \:\: \:\: \:\: \:\: \:\:\mathbf{1}^T \mathbf{s}=k  \\
& \:\: \:\: \:\: \:\: \:\: \:\: \:\: \:\: \: \: \: s_i \in \left \{ 0,1 \right \}, i=1,\cdots,n.
\end{split}
\end{equation}
where $\overline{f}$ refers to one of the asymptotic approximations for ${\rm MSE}$, ${\rm LCE}$ or the ${\rm WEV}$, that have been computed in the previous section. We present in the sequel two different methods. The first one, termed blind convex relaxation technique relies on the use of the convex relaxation approach used in \cite{boyd}, while the second one, is merely based on the use of a greedy algorithm that solves the Problem in  \eqref{eq:blind}
\begin{enumerate}
    \item Blind convex relaxation technique:
        This method builds upon the use of the convex relaxation concept. It replaces the boolean constraints in \eqref{eq:blind} by the convex constraints $0\leq s_i\leq 1$. In doing so, we obtain the following optimization problem:
\begin{equation}
    \label{eq:blind_cvx}
\begin{split}
    &{\bf s}^*= \underset{\mathbf{s}}{\text{argmin}} \overline{f}({\bf s}) \\ 
& \text{s.t.} \:\: \:\: \:\: \:\: \:\: \:\: \:\:\mathbf{1}^T \mathbf{s}=k  \\
& \:\: \:\: \:\: \:\: \:\: \:\: \:\: \:\: \: \: \: 0\leq s_i\leq 1 , i=1,\cdots,n.
\end{split}
\end{equation}
It is worth mentioning that, although the non-convex relaxation techniques are now replaced by the convex ones $0\leq s_i\leq 1$, it is not clear whether the obtained problem \eqref{eq:blind_cvx} remains convex. This is because we are not sure whether the objective function, representing the almost sure deterministic equivalent of one of the error measures is still \eqref{eq:blind_cvx} convex. The following theorem answers this question and establishes the convexity of the deterministic
approximation of the ${\rm MSE}$, ${\rm LCE}$ and ${\rm WEV}$. 
\begin{theorem}
    Define ${\overline{f}}$ as:
    \begin{align*}
        {\overline{f}}&:\mathbb{R}_{+}^n\to \mathbb{R}_{+}\\
                      &{\bf s}\mapsto \overline{f}({\bf s})
    \end{align*}
    where $\overline{f}$ is either $\overline{\rm MSE}({\bf s})$, $\overline{\rm LCE}({\bf s})$ or $\overline{\rm WEV}({\bf s})$. Then, $\overline{f}$ is convex in $\mathbb{R}_{+}^{n}$.
    \label{th:convexity}
\end{theorem}
\begin{proof}
    See Appendix
\end{proof}
Based on this Theorem, one can thus resort to standard convex tools to solve Problem \eqref{eq:blind_cvx}. This step often requires the knowldege of the gradient of $\overline{f}$ with respect to ${\bf s}$ which can be found in the top of the next page where  $\mathbf{R}_{\mathbf{s}} = \mathbf{R}^{\frac{1}{2}} \diag\left(\mathbf{s}\right)\mathbf{R}^{\frac{1}{2}}$ (A proof can be found in Appendix A)
\begin{figure*}
	\begin{align*}
	 	\frac{\partial \overline{\text{MSE}}\left(\mathbf{s}\right)}{\partial s_i}& =\frac{\partial \delta}{\partial s_i}=\delta_i'=-
	 	\frac{\delta  \left[\mathbf{R}^{\frac{1}{2}}\left(\mathbf{I}+\delta\mathbf{R}_s\right)^{-2}\mathbf{R}^{\frac{1}{2}}\right]_{i,i}}{ \tr\left[\mathbf{R}_{\mathbf{s}}\left(\mathbf{I}+\delta \mathbf{R}_{\mathbf{s}}\right)^{-2}\right]}. \\
	 		\frac{\partial \overline{\text{LCE}}\left(\mathbf{s}\right)}{\partial s_i}
	 		& 	= \left(1-c\right)\frac{\delta_i'}{\delta} + \frac{\delta_i'}{m \delta}\tr\left(\mathbf{I}+\delta \mathbf{R}_s\right)^{-1} 
	 		 - \frac{\delta}{m}	\left[\mathbf{R}^{\frac{1}{2}}\left(\mathbf{I}+\delta\mathbf{R}_s\right)^{-1}\mathbf{R}^{\frac{1}{2}}\right]_{i,i}. \\
	 		\frac{\partial \eta}{\partial s_i} & = \eta'_i  = -\eta  \left( \tr\left[\mathbf{R}_{\mathbf{s}}\left(\mathbf{I}+\eta \mathbf{R}_{\mathbf{s}}\right)^{-2}\right] - \tr\left[\mathbf{R}_{\mathbf{s}}\left(\mathbf{I}+\eta \mathbf{R}_{\mathbf{s}}\right)^{-3}\right]\right)^{-1}  \left( \left[\mathbf{R}^{\frac{1}{2}}\left(\mathbf{I}+\eta\mathbf{R}_s\right)^{-2}\mathbf{R}^{\frac{1}{2}}\right]_{i,i}- \left[\mathbf{R}^{\frac{1}{2}}\left(\mathbf{I}+\eta\mathbf{R}_s\right)^{-3}\mathbf{R}^{\frac{1}{2}}\right]_{i,i}\right). \\
	 		 \frac{\partial \overline{\text{WEV}}\left(\mathbf{s}\right)}{\partial s_i} & =-\frac{\left(c-1\right)\eta_i'}{\eta^2}+\frac{\eta_i'}{m\eta^2}  \left[2\tr\left(\mathbf{I}+\eta\mathbf{R}_s\right)^{-1}-\tr\left(\mathbf{I}+\eta\mathbf{R}_s\right)^{-2}\right] +\frac{1}{m}\left[\mathbf{R}^{\frac{1}{2}}\left(\mathbf{I}+\eta\mathbf{R}_s\right)^{-2}\mathbf{R}^{\frac{1}{2}}\right]_{i,i}.
	\end{align*}
\end{figure*}
The solution  of \eqref{eq:blind_cvx} is a vector with real positive values not necessarily zeros and ones. To obtain the indexes of the selected measurements, we should order its entries and then set the greatest ones to 1 and set the remaining to zero. \par 
    \item Greedy blind algorithms:
        Greedy algorithms have been widely applied to the framework of wireless communications, particularly in scheduling where the aim is to select the set of users that maximizes a certain utility function \cite{sigdel}. The use of greedy algorithms for measurements selection is, however, less common. 
        In order to stress  the wide scope of applicability of the proposed algorithm, we consider here the problem of  selecting the index of measurements that minimizes a pre-defined error measure $f\left(\mathcal{H},\mathcal{S}\right)$, where $\mathcal{H}$ is some information about the measurement matrix $\mathbf{H}$ \footnote{$\mathcal{H}$ could be for example the statistics given by $\mathbf{R}$ or the full channel matrix $\mathbf{H}$.} and $\mathcal{S}$ is a set of $k$ indexes from
        $\left\{1,\cdots n\right\}$. The principle of the proposed greedy algorithm is as follows. First, we start by choosing an initial candidate set $\mathcal{S}$ obtained by randomly selecting a pattern (set of measurements indexes) of size $k$. Then, select from the set of the remaining indexes ($\overline{\mathcal{S}}=\left\{1,\cdots,n\right\}\backslash\mathcal{S}$), the first value that, when replaced with one of the indexes in $\mathcal{S}$ leads to a reduction in
$f\left(\mathcal{H},\mathcal{S}\right)$. When this occurs,  $\mathcal{S}$ is updated by replacing the index that presents the largest reduction in $f\left(\mathcal{H},\mathcal{S}\right)$. 
This procedure is repeated for a predetermined number of iterations $K$. The corresponding algorithm is detailed in Algorithm 1. It can be applied for any metric 
$f\left(\mathcal{H},\mathcal{S}\right)$. This implies  that the greedy algorithm might be considered as a channel aware algorithm when $\mathcal{H}$ is given by ${\bf H}$ and entirely blind when ${\mathcal{H}}$ contains only statistical information about the channel.

\begin{proposition} 
\label{convrgence_greedy}
The greedy algorithm described by the steps of Algorithm 1 is guaranteed to converge.
\end{proposition}
\begin{proof}
By construction of Algorithm 1, the computation metric $f\left(\mathcal{H},\mathcal{S}\right)$ decreases on each iteration. Since the performance is bounded by the optimal performance achieved through the 
exhaustive search, the algorithm produces a decreasing bounded sequence, which implying its convergence.   
\end{proof}
    \end{enumerate}
\subsection{Complexity Analysis}
In this part, we discuss the complexity of the different selection algorithms. Consider first the case of full blind methods. The convex approach requires $\mathcal{O}(n^3)$ operations which is the cost of using interior-point methods. As for the greedy approach, the computational complexity is governed by two factors.   i) the complexity needed at every iteration and ii) the total number of iterations until convergence, which  we denote by $K$. At every iteration, we need to perform  $k \left(n-k\right)$ computations,  thus in total, we
need $K \times k \left(n-k\right) $ computations. As $k$ and $n$ are assumed to be commensurable, the computation complexity is thus $K \times \mathcal{O}\left(n^2\right)$.
Now, if the greedy approach and convex relaxation based techniques are applied when the channel is perfectly known, complexity has to be multiplied by $N$ which represents the number of times over which the channel changes. To sum up,  we present the complexity achieved by the proposed selection algorithms in Table \ref{table:complexity} in both full blind and CSI aware scenarios. 
Figure \ref{fig:iterations_conv} represents the MSE performance achieved by the greedy algorithm for both cases (channel-aware and blind) as a function of the number of iterations. For both cases, the greedy algorithm requires a number of iterations, $K = 2$ to converge. This value of $K$ will be implemented in all the next simulations for the greedy algorithm.
\begin{table}
\begin{center}
	\begin{tabular}{|c||c|}
		\hline
		\rowcolor{lightgray}
		Algorithm & Complexity  \\\hline
		Convex Optimization(Channel-aware) & $N \times \mathcal{O}\left(n^3\right)$ \\\hline
		Convex Optimization(Blind) & $\mathcal{O}\left(n^3\right)$ \\\hline
		Greedy(Channel-aware)& $K \times N\times \mathcal{O}\left(n^2\right)$ \\\hline
		Greedy(Blind)  & $K \times \mathcal{O}\left(n^2\right)$ \\
		\hline
	\end{tabular}
\end{center}	
\caption{Computational complexity of the different proposed algorithms.}
\label{table:complexity}
\end{table}
\begin{figure}[h!]
	\centering
	\includegraphics[scale=0.65]{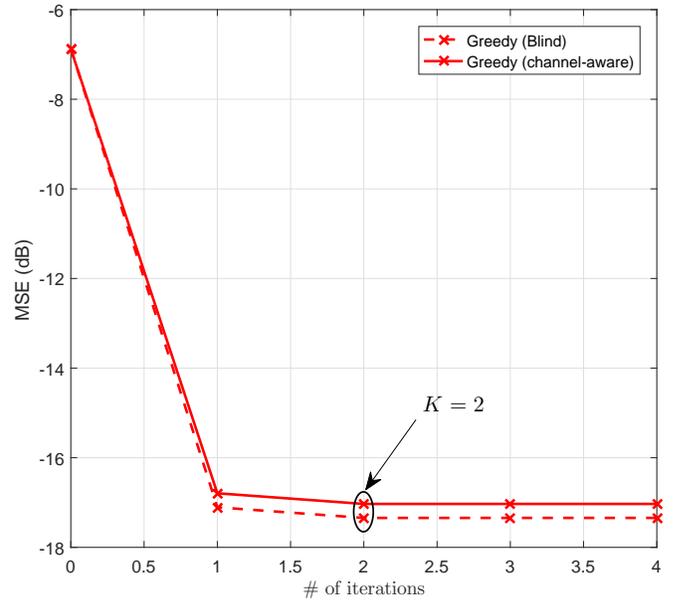}
	\caption{Average MSE performance (over 100 realizations) of the greedy approach for both
		cases (channel-aware and blind) as function of the number of
		iterations. $\mathbf{R}$ follows the model in (\ref{spatial_corr})  ($m=30$, $n = 100$, $k = 50$ and $d = 2$).}
	\label{fig:iterations_conv}
\end{figure} 
\section{Potential Applications} \label{application}
\label{application_section}
In this section, we show two potential applications in which the proposed blind measurement selection algorithms can be applied to help reducing the computational cost.  
The first application concerns antenna selection in massive MIMO systems while the second focuses on the problem of sensor selection in WSN. In the following, we provide a detailed description of the system model in each application and analyze the performance in terms of the error measures proposed in section \ref{formulation}.
\subsection{Antenna Selection for Single-cell Uplink Massive MIMO Systems}
Consider the uplink of a single cell MU-MIMO system in which $m$ single-antenna users are served by a single base station (BS) equipped with $n$ antennas with $m < n$, as sketched in Figure \ref{fig:model}. Assuming that the users' signals are perfectly synchronized in time and frequency, the received vector at the BS is given by
\begin{equation}
\label{received}
\mathbf{y}=\sqrt{\rho} \mathbf{H}\mathbf{x}+\mathbf{e},
\end{equation}
where $\mathbf{y}\in \mathbb{C}^{n \times 1}$ is the received vector at the BS, $\rho$ is the average transmit power per user and $\mathbf{x} \in \mathbb{C}^{m \times 1}$ is the data vector.  Matrix $\mathbf{H}=\left \{ h_{i,j} \right \} \in \mathbb{C}^{n \times m}$ denotes the narrow-band uplink channel matrix where $h_{i,j}$ is the channel coefficient between the $j$-th user and the $i$-th BS's antenna. Moreover, we assume that the random channel $\mathbf{H}$ exhibits the one-sided Kronecker model given by
\begin{equation}
\label{kronecker}
\mathbf{H}=\mathbf{R}^{\frac{1}{2}} \mathbf{W},
\end{equation}
where $\mathbf{W} \in \mathbb{C}^{n\times m}$ is a matrix with $i.i.d$ circularly symmetric zero mean unit-variance complex Gaussian entries, $\mathbf{R}$ models the spatial receive correlation matrix, whose elements represent  the correlation between the antennas of the BS and ${\bf e}$ denotes noise vector at the BS with $i.i.d$ circularly symmetric zero mean unit-variance complex Gaussian entries, i.e., $\mathbf{e} \sim \mathcal{CN}\left(\mathbf{0},\mathbf{I}_n\right)$. 
At the receiver side, the BS  estimates the transmitted vector $\mathbf{x}$ using ${\bf y}$. Several detection procedures can be used, among which are the optimal maximum likelihood (ML) detector and the least squares. The latter  achieves a good balance between complexity and performance. In communication parlance, it is referred to as zero-forcing (ZF) detection and is given by
\begin{align*}
\widehat{\mathbf{x}} = \frac{1}{\sqrt{\rho}}\mathbf{H}^{\dagger} \mathbf{y},
\end{align*}
where $\mathbf{H}^{\dagger}=\left(\mathbf{H}^H\mathbf{H}\right)^{-1}\mathbf{H}^H$ is the pseudo-inverse of $\mathbf{H}$.
Even  with the use of a ZF detector in place of the optimal ML decoder, the complexity of the decoding might be prohibitively high as a result of the high number of antennas $n$. Antenna selection appears thus as a valuable technique that can allow decoding with a lower complexity. In this respect, we evaluate the performance of the aforementioned antenna selection procedures for this practical scenario. 
\begin{figure}[h!]
	\centering
	\includegraphics[width=3.5in]{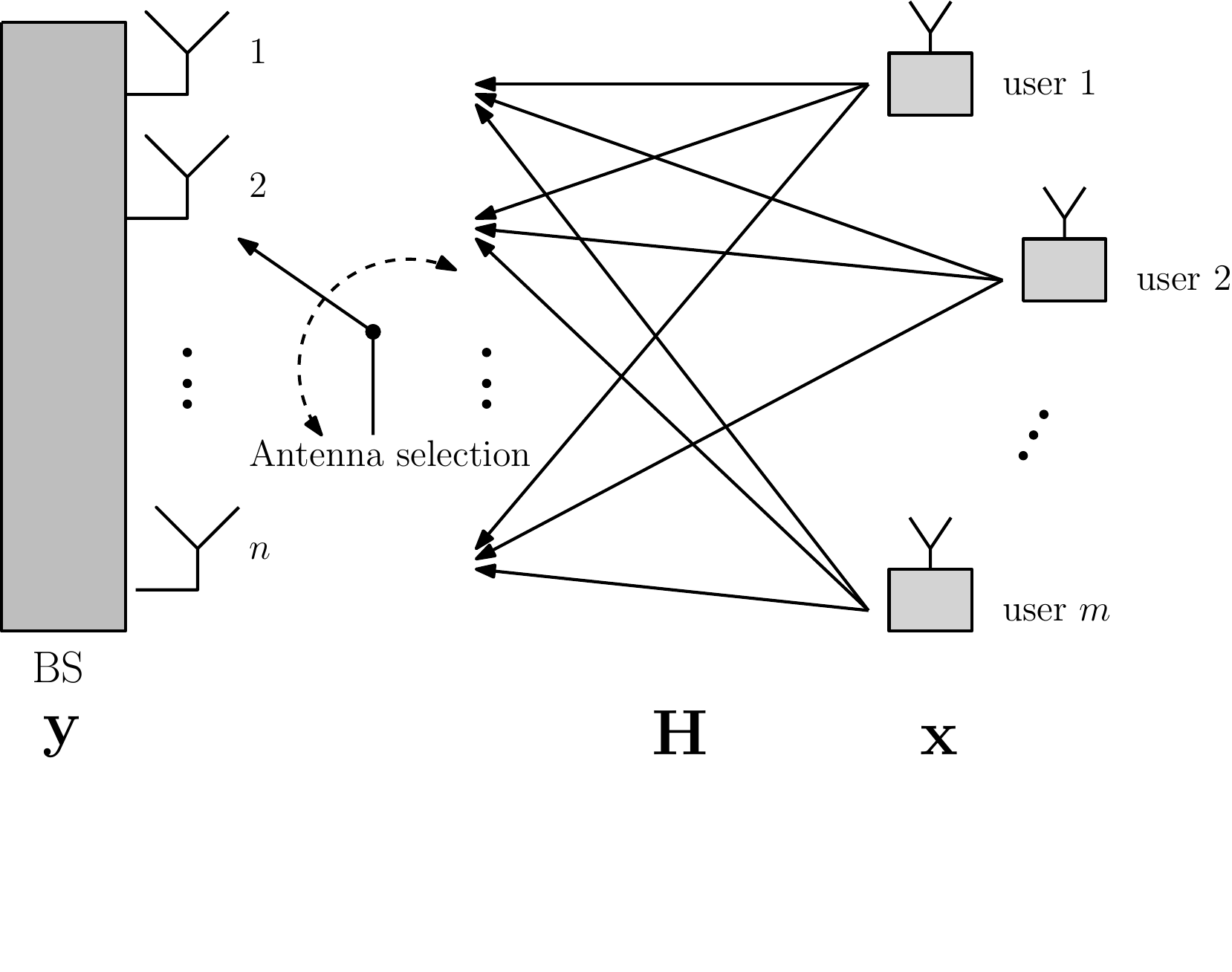}
	\caption{System model of an uplink MU-MIMO system composed of a BS equipped with $n$ antennas and serving $m$ single-antenna users.}
	\label{fig:model}
\end{figure}
\begin{figure}[h!]
	\centering
	\includegraphics[scale=0.4]{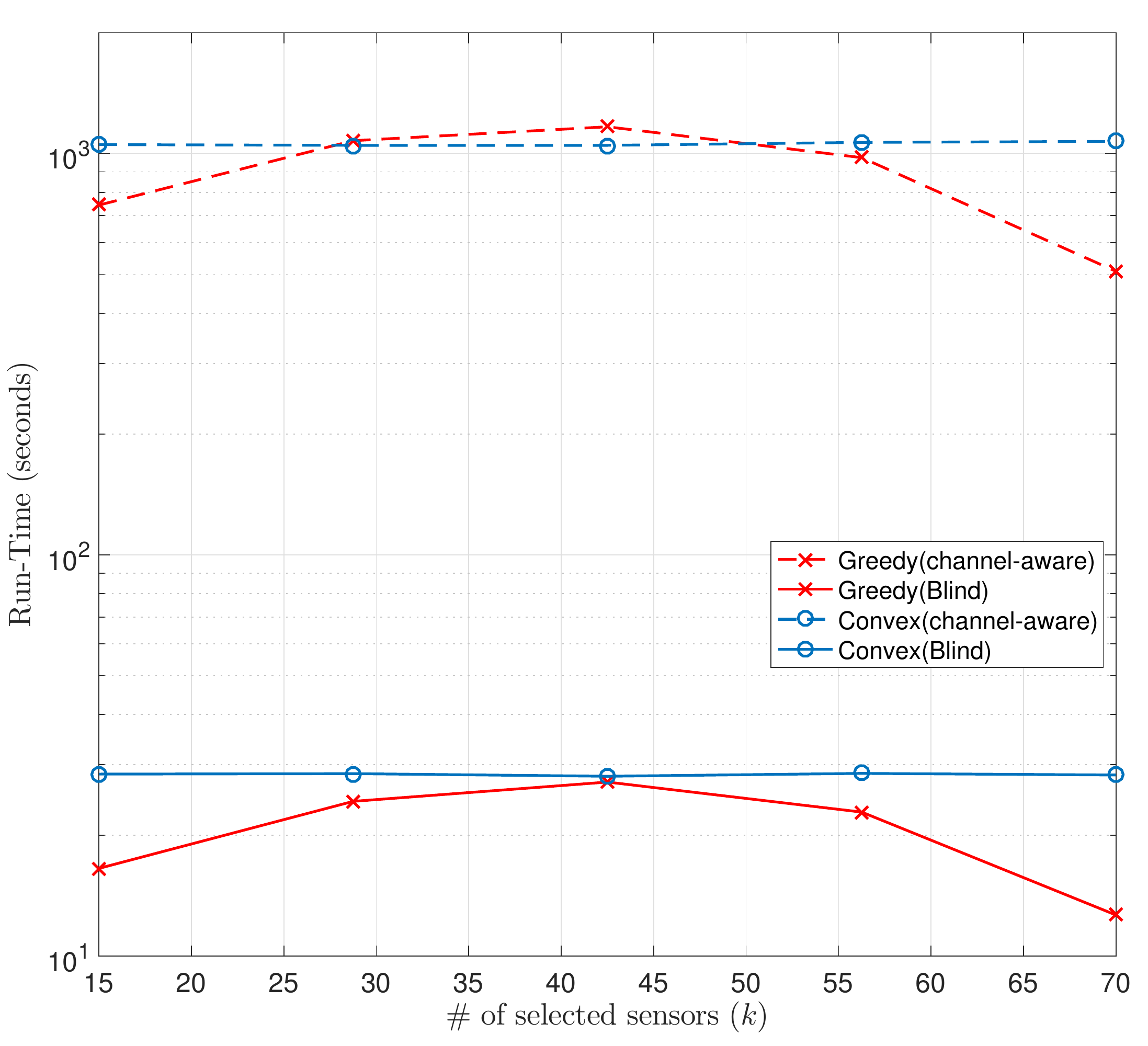}
	\caption{Run-time in seconds vs the number of selected
		antennas $k$ for the considered selection algorithms.}
	\label{fig:Time_MIMO}
\end{figure}
\subsubsection*{Numerical Example} 
In all experiments, we assume that the number of users $m$ is  $30$ and the total budget of antennas is $n=100$. We also set the ${\rm SNR}$ to $\rho=20$ dB. The following spatial correlation model \cite{giusi} is considered:  
\begin{equation}
\label{spatial_corr}
\mathbf{R}_{i,j} = \exp\left(-0.05. d^2 \left(i-j\right)^2\right), \: 1 \leq i,j \leq n.
\end{equation}
that models a broadside Gaussian power azimuth spectrum with a  root-mean-square spread of $2^{\circ}$ where $d$ corresponds to the  antenna separation in wavelength units. The greedy and the convex relaxation based algorithms are considered in channel-aware ($\mathcal{H}=\mathbf{H}$) and fully blind scenarios ($\mathcal{H}=\mathbf{R}$).
Figure \ref{fig:MSE_MIMO} reports the achieved averaged MSE (over $100$ realizations) for all proposed selection algorithms along with the random selection algorithm that randomly picks a set of $k$ antennas out of $n$. As a major observation, we note that when the correlation between antennas is low ($d=4$), the proposed blind algorithms are not that advantageous as compared to the random selection algorithm. This is kind of expected since the rows of ${\bf H}$ become almost statistically independent and
identically distributed. They are thus statistically equivalent and selecting any $k$ rows would asymptotically achieve the same MSE, as can be evidenced from the deterministic equivalent of the MSE shown in Lemma 1. However, with  the impact of correlation becoming more important $(d\downarrow)$, the gain of blind approaches over the random selection approach increases. They constitute thus a valuable option, given the fact that they only entail a loss of a up to $1$ dB  as
compared to channel-aware algorithms. This can be clearly seen in Figure \ref{fig:MSE_d}, where we plot the average MSE against the antennas' separation $d$. It is worth mentioning that for $d=1$ and $d=2$, the proposed blind greedy approach outperforms the channel-aware convex approach. This may sounds counter intuitive, but this is in fact due to the quantization effect at the output of the relaxed optimization problem.
Also, it is worth mentioning that the blind greedy algorithms perform antenna selection at the pace of the variation of the large scale statistics. This must be compared with the channel aware algorithms which are required to perform antenna selection for every channel realization. A high reduction in the computational complexity is thus achieved as evidenced by  Figure \ref{fig:Time_MIMO}, showing the run-time in seconds consumed by the different selection
algorithms. 
All in all, it appears that the proposed blind selection techniques present in reality a better trade-off between complexity and performance.
\begin{figure*}[t!]
	\centering
	\includegraphics[scale=0.6]{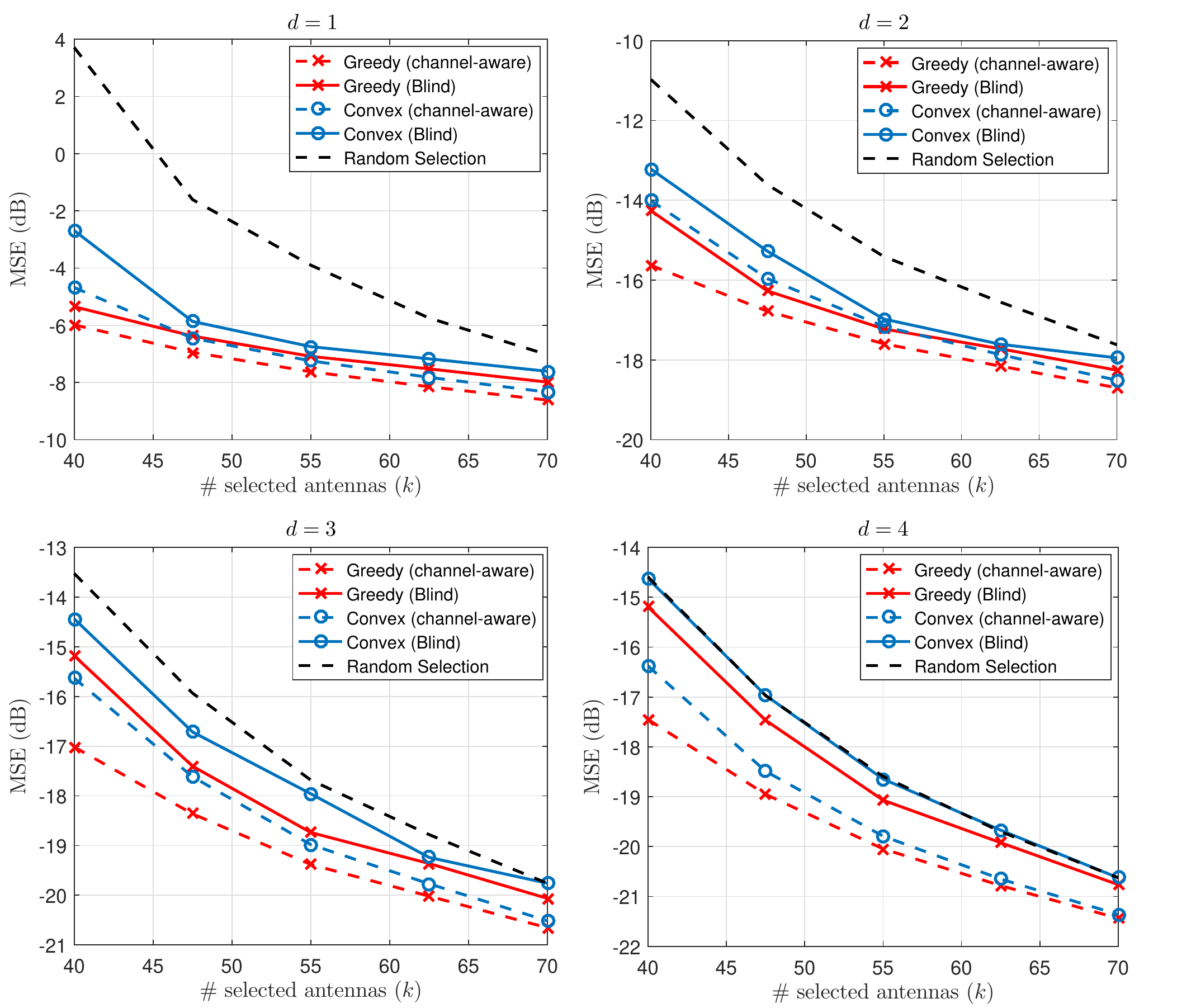}
	\caption{Average MSE achieved by the proposed selection techniques versus $k$ for different values of the antennas' separation $d$.}
	\label{fig:MSE_MIMO}
\end{figure*}
\begin{figure}[t!]
	\centering
	\includegraphics[scale=0.66]{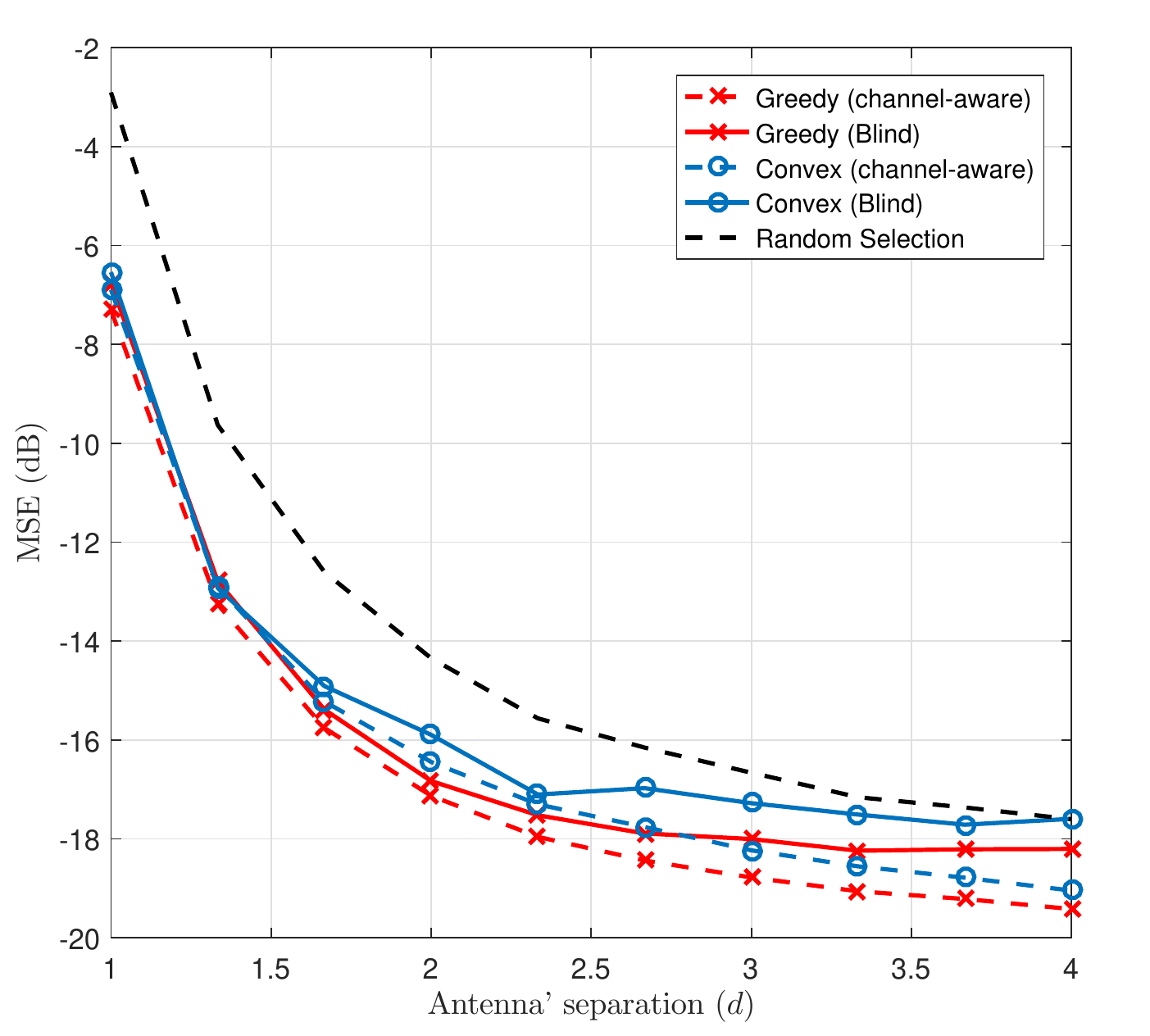}
	\caption{Average MSE achieved by the proposed selection techniques versus the antennas' separation $d$ with $k=50$.}
	\label{fig:MSE_d}
\end{figure}
\subsection{Sensor Selection in WSN}
We consider a wireless sensor network (WSN) with total number of $n=100$ sensor nodes sensing a phenomena of dimension $m=30$ where sensors are randomly deployed over a circular area of radius $30$m. Thus, we have the same linear system as in (\ref{linear1}). Unlike the previous application, we assume that the rows of $\mathbf{H}$ are statistically independent and consider the correlation in the measurement noise. The error covariance matrix in this case becomes
\begin{equation}
\label{sigma_wsn}
\mathbf{\Sigma} = \left(\mathbf{W}^H \mathbf{\Phi}^{-1}\mathbf{W}\right)^{-1},
\end{equation} 
where $\mathbf{\Phi}$ is the noise covariance matrix given by \cite{noise_cov}
\begin{equation}
\label{corr_sensors}
\mathbf{\Phi}_{i,j} = \sigma^2 \exp\left(-\rho \left \| S_i-S_j \right \|_2\right), \: 1 \leq i,j \leq n,
\end{equation}
where $\left \| S_i-S_j \right \|_2$ denotes the Euclidean distance between nodes' locations in the 2D plane $S_i$ and $S_j$, $\sigma^2=1$ and $\rho$ is the correlation parameter that controls the strength of the spatial correlation. Obviously, a larger $\rho$ results in a weaker correlation and vice-versa. As shown in (\ref{sigma_wsn}), $\mathbf{\Sigma}$ is function of $\mathbf{\Phi}^{-1}$ which presents a slight difference as compared to the previous application, where we use $\mathbf{\Phi}^{-1}$ instead of $\mathbf{R}$ \footnote{All the derived results concerning the asymptotic equivalents and their convergence can be applied in a straightforward manner to the case $\mathbf{R} = \mathbf{\Phi}^{-1}$.}.

\par 
We examine the performance in terms of the MSE as well as the LCE and the WEV. 
As shown in Figure \ref{fig:sensors_perform1}, similar observations to the application of massive MIMO can be conducted. We notice that increasing the correlation ($\rho \downarrow$) results in a better performance of the proposed blind selection algorithms as compared to the random selection algorithm and vice versa. We also observe that the proposed blind convex approach fails to achieve a good performance for both the MSE and the WEV and yields a performance that is worse than that of random selection. This is due to the same quantization effect as explained in the case of massive MIMO.  In Figures \ref{fig:LCE_perform1} and \ref{fig:WEV_perform1}, we show the performance of the $\text{LCE}$ and the $\text{WEV}$ against the correlation parameter $\rho$ respectively. Figures \ref{fig:LCE_perform1} and \ref{fig:WEV_perform1} clearly show that the performance for both the $\text{LCE}$ and the $\text{WEV}$ improves with increasing the correlation ($\rho$ $\downarrow$ ) for the proposed blind approach. This is suitable in such application, since a central node can perform sensor selection without knowledge of the channel matrix which may require a huge overhead.



\begin{figure*}[t!]
	\centering
\input{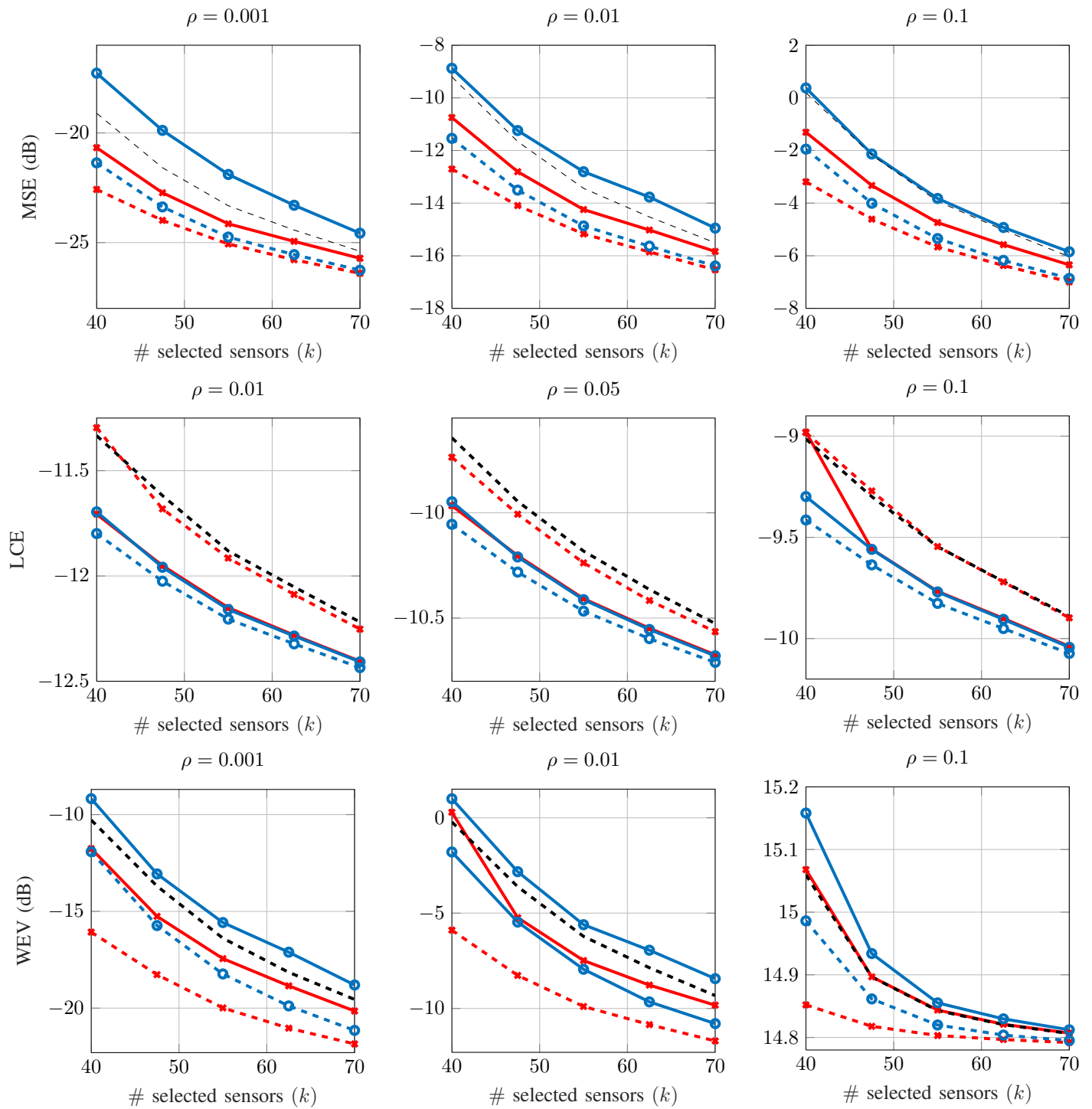}
\caption{Average error $\text{MSE}$, $\text{LCE}$ and $\text{WEV}$) achieved by the proposed selection techniques for different values of the correlation parameter $\rho$ in the WSN scenario.}
	\label{fig:sensors_perform1}
\end{figure*}
\begin{figure}[t!]
	\centering
%
%
\definecolor{mycolor1}{rgb}{0.00000,0.44706,0.74118}%
\begin{tikzpicture}

\begin{axis}[%
width=3in,
height=2.8in,
at={(0in,0.771in)},
scale only axis,
xmode=log,
xmin=0.001,
xmax=0.1,
xminorticks=true,
xlabel style={font=\color{white!15!black}},
xlabel={Corrleation parameter $\rho$},
ymin=-15,
ymax=-9,
ylabel style={font=\color{white!15!black}},
ylabel={LCE},
axis background/.style={fill=white},
xmajorgrids,
xminorgrids,
ymajorgrids,
legend style={at={(0.31,0)}, anchor=south west, legend cell align=left, align=left, draw=white!15!black}
]
\addplot [color=red, dashed, line width=1.5pt, mark=x, mark options={solid, red}]
  table[row sep=crcr]{%
0.001	-13.8660\\
0.012	-11.3784\\
0.023	-10.7257\\
0.034	-10.3472\\
0.045	-10.0662\\
0.056	-9.8520\\
0.067	-9.6830\\
0.078	-9.5377\\
0.089	-9.4096\\
0.1	    -9.3130\\
};
\addlegendentry{Greedy (channel-aware)}

\addplot [color=red, line width=1.5pt, mark=x, mark options={solid, fill=red, red}]
  table[row sep=crcr]{%
0.001	-14.0465\\
0.012	-11.5646\\
0.023	-10.9078\\
0.034	-10.5188\\
0.045	-10.2488\\
0.056	-10.0303\\
0.067	-9.8552\\
0.078	-9.7067\\
0.089	-9.5750\\
0.1	-9.4699\\
};
\addlegendentry{Greedy (Blind)}

\addplot [color=mycolor1, dashed, line width=1.5pt, mark=o, mark options={solid, mycolor1}]
  table[row sep=crcr]{%
0.0010  -14.1206\\
0.0120  -11.6305\\
0.0230  -10.9797\\
0.0340  -10.5905\\
0.0450  -10.3176\\
0.0560  -10.1009\\
0.0670  -9.9270\\
0.0780  -9.7759\\
0.0890  -9.6498\\
0.1000  -9.5417\\
};
\addlegendentry{Convex (channel-aware)}

\addplot [color=mycolor1, line width=1.5pt, mark=o, mark options={solid, mycolor1}]
  table[row sep=crcr]{%
0.0010	-14.0526\\
0.0120  -11.5674\\
0.0230  -10.9120\\
0.0340  -10.5208\\
0.0450  -10.2492\\
0.0560  -10.0358\\
0.0670  -9.8571\\
0.0780  -9.7066\\
0.0890  -9.5801\\
0.1000  -9.4709\\
};
\addlegendentry{Convex (Blind)}

\addplot [color=black, dashed, line width=1.5pt]
  table[row sep=crcr]{%
0.0010  -13.8085\\
0.0120  -11.3303\\
0.0230  -10.6874\\
0.0340  -10.2851\\
0.0450	-10.0148\\
0.0560  -9.8050\\
0.0670  -9.6394\\
0.0780  -9.4885\\
0.0890  -9.3584\\
0.1000  -9.2588\\
};
\addlegendentry{Random Selection}

\end{axis}
\end{tikzpicture}%
	\caption{Average LCE achieved by the proposed selection techniques versus the correlation parameter $\rho$ in the case of WSN with $k=50$.}
	\label{fig:LCE_perform1}
\end{figure}
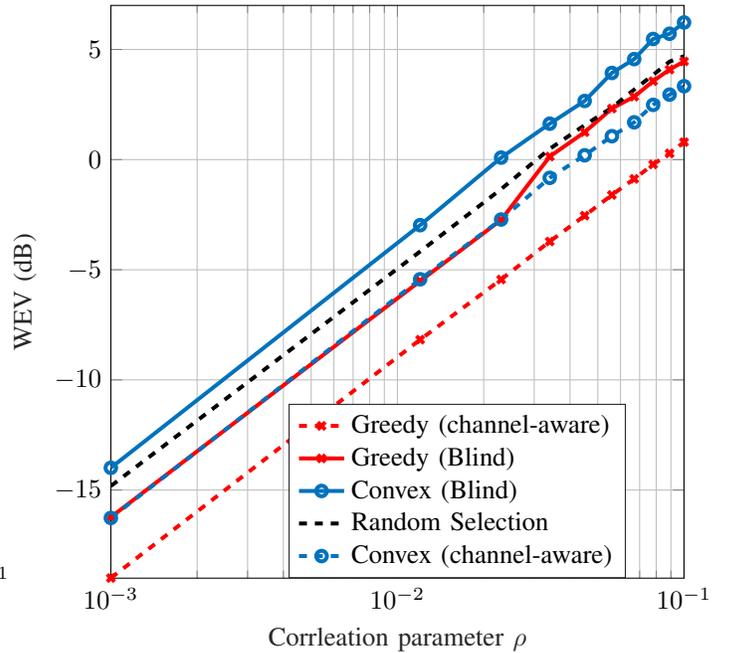
\begin{figure}[t!]
	\label{blueperformance}
	\centering
%
%
\definecolor{mycolor1}{rgb}{0.00000,0.44706,0.74118}%
\begin{tikzpicture}

\begin{axis}[%
width=3in,
height=3in,
at={(0.696in,0.903in)},
scale only axis,
xmode=log,
xmin=0.001,
xmax=0.1,
xminorticks=true,
xlabel style={font=\color{white!15!black}},
xlabel={Corrleation parameter $\rho$},
ymin=-19,
ymax=7,
ylabel style={font=\color{white!15!black}},
ylabel={WEV (dB)},
axis background/.style={fill=white},
xmajorgrids,
xminorgrids,
ymajorgrids,
legend style={at={(0.31,0)}, anchor=south west, legend cell align=left, align=left, draw=white!15!black}
]
\addplot [color=red, dashed, line width=1.5pt, mark=x, mark options={solid, red}]
  table[row sep=crcr]{%
 0.0010  -18.9939\\
 0.0120   -8.1732\\
 0.0230   -5.4435\\
 0.0340   -3.7171\\
 0.0450   -2.5450\\
 0.0560   -1.6048\\
 0.0670   -0.8748\\
 0.0780   -0.2192\\
 0.0890    0.2825\\
 0.1000    0.7952\\
};
\addlegendentry{Greedy (channel-aware)}

\addplot [color=red, line width=1.5pt, mark=x, mark options={solid, red}]
  table[row sep=crcr]{%
    0.0010  -16.2318\\
    0.0120   -5.5108\\
    0.0230   -2.7379\\
    0.0340    0.1390\\
    0.0450    1.2448\\
    0.0560    2.3180\\
    0.0670    2.8584\\
    0.0780    3.5604\\
    0.0890    4.0881\\
    0.1000    4.4564\\
};
\addlegendentry{Greedy (Blind)}

\addplot [color=mycolor1, line width=1.5pt, mark=o, mark options={solid, mycolor1}]
  table[row sep=crcr]{%
    0.0010  -13.9916\\
    0.0120   -2.9757\\
    0.0230    0.0900\\
    0.0340    1.6311\\
    0.0450    2.6626\\
    0.0560    3.9345\\
    0.0670    4.5630\\
    0.0780    5.4768\\
    0.0890    5.7130\\
    0.1000    6.2248\\
};
\addlegendentry{Convex (Blind)}

\addplot [color=black, dashed, line width=1.5pt]
  table[row sep=crcr]{%
    0.0010  -14.8169\\
    0.0120   -4.1693\\
    0.0230   -1.3393\\
    0.0340    0.4889\\
    0.0450    1.5639\\
    0.0560    2.3591\\
    0.0670    3.1792\\
    0.0780    3.8441\\
    0.0890    4.4492\\
    0.1000    4.6942\\
};
\addlegendentry{Random Selection}

\addplot [color=mycolor1, dashed, line width=1.5pt, mark=o, mark options={solid, mycolor1}]
  table[row sep=crcr]{%
    0.0010  -16.2675\\
    0.0120   -5.4314\\
    0.0230   -2.7191\\
    0.0340   -0.8277\\
    0.0450    0.1955\\
    0.0560    1.0635\\
    0.0670    1.6929\\
    0.0780    2.4848\\
    0.0890    2.9500\\
    0.1000    3.3332\\
};
\addlegendentry{Convex (channel-aware)}

\end{axis}
\end{tikzpicture}%
	\caption{Average WEV in dB achieved by the proposed selection techniques versus the correlation parameter $\rho$ in the case of WSN with $k=50$.}
	\label{fig:WEV_perform1}
\end{figure}
\section{Conclusion} \label{conclusion}
In this paper, we introduced blind techniques for measurement selection. In particular, we showed that using tools from random matrix theory, it is possible to asymptotically approximate error measures that are commonly used in this context. As such, perfect knowledge of the measurement matrix is not needed and only statistics are required to perform measurement selection. We proposed two techniques: the first is based on a greedy approach and the second is based on a convex relaxation heuristic. The proposed blind selection techniques have been tested in two applications related to wireless communications: the first is antenna selection in uplink multiusers massive MIMO systems and the second is sensor selection in wireless sensor networks. Numerical results showed that the blind techniques have a comparable performance to techniques that require full knowledge of the measurement matrix, especially at high correlation.
\bibliographystyle{IEEEtran}
\bibliography{References}

\begin{thebibliography}{10}
\providecommand{\url}[1]{#1}
\csname url@samestyle\endcsname
\providecommand{\newblock}{\relax}
\providecommand{\bibinfo}[2]{#2}
\providecommand{\BIBentrySTDinterwordspacing}{\spaceskip=0pt\relax}
\providecommand{\BIBentryALTinterwordstretchfactor}{4}
\providecommand{\BIBentryALTinterwordspacing}{\spaceskip=\fontdimen2\font plus
\BIBentryALTinterwordstretchfactor\fontdimen3\font minus
  \fontdimen4\font\relax}
\providecommand{\BIBforeignlanguage}[2]{{%
\expandafter\ifx\csname l@#1\endcsname\relax
\typeout{** WARNING: IEEEtran.bst: No hyphenation pattern has been}%
\typeout{** loaded for the language `#1'. Using the pattern for}%
\typeout{** the default language instead.}%
\else
\language=\csname l@#1\endcsname
\fi
#2}}
\providecommand{\BIBdecl}{\relax}
\BIBdecl

\bibitem{boyd}
S.~Joshi and S.~Boyd, ``{S}ensor {S}election via {C}onvex {O}ptimization,''
  \emph{IEEE Trans. on Signal Processing}, vol.~57, no.~2, pp. 451--462, Feb
  2009.

\bibitem{robotics_ref}
G.~Hovland and B.~McCarragher, ``{D}ynamic {S}ensor {S}election for {R}obotic
  {S}ystems,'' in \emph{IEEE Int. Conf. Robotics Automation}, vol.~1, 1997, pp.
  272--277.

\bibitem{kammer}
D.~Kammer, ``{S}ensor {P}lacement for {O}n-{O}rbit {M}odal {I}dentification and
  {C}orrelation of {L}arge {S}pace {S}tructures,'' \emph{J. Guid., Control,
  Dynam}, vol.~14, pp. 251--259, 1991.

\bibitem{genetic}
L.~Yao, W.~Sethares, and D.~Kammer, ``{S}ensor {P}lacement for {O}n-{O}rbit
  {M}odal {I}dentification via a {G}enetic {A}lgorithm,'' \emph{Amer. Inst.
  Aeronaut. Astronaut. J.}, vol.~31, p. 1922–1928, 1993.

\bibitem{nguyen}
N.~Nguyen and A.~Miller, ``{A} {R}eview of {S}ome {E}xchange {A}lgorithms for
  {C}onstructing {D}iscrete {D}-{O}ptimal {D}esigns,'' \emph{Comput. Statist.
  Data Anal.}, vol.~14, pp. 489–--498, 1992.

\bibitem{john}
R.~John and N.~Draper, ``{D}-{O}ptimality for {R}egression {D}esigns: {A}
  {r}eview,'' \emph{Technometrics}, vol.~17, no.~1, pp. 15--–23, 1997.

\bibitem{larsson}
E.~{L}arsson, O.~{E}dfors, F.~{T}ufvesson, and T.~{M}arzetta, ``{M}assive
  {MIMO} for {N}ext {G}eneration {W}ireless {S}ystems,'' \emph{IEEE
  Communications Magazine}, vol.~52, no.~2, pp. 186--195, February 2014.

\bibitem{marzetta}
T.~L. {M}arzetta, ``{N}oncooperative {C}ellular {W}ireless with {U}nlimited
  {N}umbers of {B}ase {S}tation {A}ntennas,'' \emph{IEEE Transactions on
  Wireless Communications}, vol.~9, no.~11, pp. 3590--3600, November 2010.

\bibitem{rusek}
F.~{R}usek, D.~{P}ersson, B.~K. {L}au, E.~G. {L}arsson, T.~L. {M}arzetta,
  O.~{E}dfors, and F.~{T}ufvesson, ``{S}caling {U}p {MIMO}: {O}pportunities and
  {C}hallenges with {V}ery {L}arge {A}rrays,'' \emph{IEEE Signal Processing
  Magazine}, vol.~30, no.~1, pp. 40--60, Jan 2013.

\bibitem{gao}
X.~Gao, O.~Edfors, J.~Liu, and F.~Tufvesson, ``{A}ntenna {S}election in
  {M}easured {M}assive {MIMO} {C}hannels using {C}onvex {O}ptimization,'' in
  \emph{Globecom Workshops (GC Wkshps), 2013 IEEE}, Dec 2013, pp. 129--134.

\bibitem{molisch}
A.~F. Molisch and M.~Z. Win, ``{MIMO} {S}ystems with {A}ntenna {S}election,''
  \emph{IEEE Microwave Magazine}, vol.~5, no.~1, pp. 46--56, Mar 2004.

\bibitem{sanayei}
S.~Sanayei and A.~Nosratinia, ``{A}ntenna {S}election in {MIMO} {S}ystems,''
  \emph{IEEE Communications Magazine}, vol.~42, no.~10, pp. 68--73, Oct 2004.

\bibitem{yao}
L.~Yao, W.~A. Sethares, and D.~C. Kammer, ``{Sensor Placement for On-orbit
  Modal Identification via a Genetic Algorithm},'' \emph{Amer. Inst. Aeronaut.
  Astronaut. J.}, vol.~31, no.~10, pp. 1922--1928, 1993.

\bibitem{sayed}
A.~Sayed, \emph{Fundamentals of Adaptive Filtering}.\hskip 1em plus 0.5em minus
  0.4em\relax John Wiley \& Sons, 2003.

\bibitem{serfling}
R.~J. Serfling, \emph{Approximations Theorems of Mathematical
  Statistics}.\hskip 1em plus 0.5em minus 0.4em\relax John Wiley \& Sons, 2002.

\bibitem{silverstein}
J.~W. {S}ilverstein and Z.~D. {B}ai, ``{O}n the {E}mpirical {D}istribution of
  {E}igenvalues of a {C}lass of {L}arge {D}imensional {R}andom {M}atrices,''
  \emph{Journal of Multivariate Analysis}, vol.~54, pp. 175--192, May 2002.

\bibitem{dumont-07}
\emph{High SNR Approximations of the Capacity of MIMO Correlated Rician
  channels: A large System Approach}, 2007.

\bibitem{BAI99}
Z.~D. Bai and J.~W. Silverstein, ``{Exact Separation of Eigenvalues of Large
  Dimensional Sample Covariance Matrices},'' \emph{The Annals of Probability},
  vol.~27, no.~3, pp. 1536--1555, 1999.

\bibitem{sigdel}
S.~Sigdel, W.~A. Krzymień, and M.~Al-Shalash, ``{G}reedy and {P}rogressive
  {U}ser {S}cheduling for {C}o{MP} {W}ireless {N}etworks,'' in
  \emph{Communications (ICC), 2012 IEEE International Conference on}, June
  2012, pp. 4218--4223.

\bibitem{giusi}
G.~Alfano, A.~M. Tulino, A.~{L}ozano, and S.~{V}erdu, ``{C}apacity of {MIMO}
  {C}hannels with {O}ne-sided {C}orrelation,'' \emph{ISSSTA}, August 2004.

\bibitem{noise_cov}
S.~Liu, S.~P. Chepuri, M.~Fardad, E.~Maşazade, G.~Leus, and P.~K. Varshney,
  ``{S}ensor {S}election for {E}stimation with {C}orrelated {M}easurement
  {N}oise,'' \emph{IEEE Transactions on Signal Processing}, vol.~64, no.~13,
  pp. 3509--3522, July 2016.

\bibitem{dumont}
J.~Dumont, W.~Hachem, S.~Lasaulce, P.~Loubaton, and J.~Najim, ``{O}n the
  {C}apacity {A}chieving {C}ovariance {M}atrix for {R}ician {MIMO} {C}hannels:
  {A}n {A}symptotic {A}pproach,'' \emph{IEEE Transactions on Information
  Theory}, vol.~56, no.~3, pp. 1048--1069, March 2010.

\end{thebibliography}
\section*{Appendix A}
\section*{Gradient Derivation of the different Deterministic Equivalents}
\subsection*{Useful Lemmas}
\subsubsection*{Inversion Lemma}
For $\mathbf{A}$ an invertible square matrix and column vector $\mathbf{u}$ such  that $1+\mathbf{u}^T \mathbf{A}^{-1}\mathbf{u} \neq 0$, we have
\begin{align}
\left(\mathbf{A}+\mathbf{u}\mathbf{u}^T\right)^{-1} & = \mathbf{A}^{-1}-\frac{1}{1+\mathbf{u}^T\mathbf{A}^{-1}\mathbf{u}}\mathbf{A}^{-1}\mathbf{u}\mathbf{u}^T\mathbf{A}^{-1}. \\
\det \left(\mathbf{A}+\mathbf{u}\mathbf{u}^T\right) & = \det \left(\mathbf{A}\right) \left(1+\mathbf{u}^T\mathbf{A}^{-1}\mathbf{u}\right).
\end{align}

\subsubsection*{Implicit Function Theorem}
Let $f: U \times V \rightarrow \mathbb{R}$ be a continuously differentiable function and let $g: U \rightarrow V$ be the implicit function defined as follows
\begin{equation}
\left \{ \left(\mathbf{x},g\left(\mathbf{x}\right)\right) | \mathbf{x} \in U \right \} = \left \{ \left(\mathbf{x},\mathbf{y}\right)\in U \times V | f\left(\mathbf{x,y}\right)=\mathbf{c} \right \}.
\end{equation}
Then, $g$ is continuously differentiable and  
\begin{equation}
\frac{\partial g}{\partial x_i}\left(\mathbf{x}\right) = - \left(\frac{\partial f}{\partial y}\left(\mathbf{x},g\left(\mathbf{x}\right)\right)\right)^{-1} \frac{\partial f}{\partial x_i}\left(\mathbf{x},g\left(\mathbf{x}\right)\right).
\end{equation}
\subsection*{Some Useful Notations}
\begin{itemize}
	\item $\mathbf{R}_{\mathbf{s}} = \mathbf{R}^{\frac{1}{2}} \diag\left(\mathbf{s}\right)\mathbf{R}^{\frac{1}{2}}$.
	\item $\mathbf{r}_i$ is the $i$th column of $\mathbf{R}^{\frac{1}{2}}$.
	\item $\mathbf{\Psi}_i =\mathbf{R}_{\mathbf{s}} - s_i \mathbf{r}_i \mathbf{r}_i^T. $
\end{itemize}
\subsection*{Gradient of the MSE}
For the $\text{MSE}$, we have the following fixed-point equation in terms of $\delta$
\begin{align*}
\delta \tr\left[\mathbf{R}_{\mathbf{s}}\left(\mathbf{I}+\delta \mathbf{R}_{\mathbf{s}}\right)^{-1}\right]-m =0.
\end{align*}
Define
\begin{align*}
f\left(\mathbf{s},y\right) & = y \tr\left[\mathbf{R}_{\mathbf{s}}\left(\mathbf{I}+y \mathbf{R}_{\mathbf{s}}\right)^{-1}\right]-m \\
& = n-m-\tr\left[\left(\mathbf{I}+y \mathbf{R}_{\mathbf{s}}\right)^{-1}\right].
\end{align*}
Then, 
\begin{align*}
\frac{\partial f}{\partial y}\left(\mathbf{s},y\right) = \tr\left[\mathbf{R}_{\mathbf{s}}\left(\mathbf{I}+y \mathbf{R}_{\mathbf{s}}\right)^{-2}\right].
\end{align*}
On the other hand
\begin{align*}
\tr\left[\left(\mathbf{I}+y \mathbf{R}_{\mathbf{s}}\right)^{-1}\right] & = \tr\left(\mathbf{I}+y\mathbf{\Psi}_i+y s_i \mathbf{r}_i\mathbf{r}_i^T\right)^{-1} \\
& = \tr\left(\mathbf{I}+y\mathbf{\Psi}_i\right)^{-1}- \frac{y s_i \mathbf{r}_i^T \left(\mathbf{I}+y\mathbf{\Psi}_i\right)^{-2}\mathbf{r}_i}{1+ys_i\mathbf{r}_i^T\left(\mathbf{I}+y\mathbf{\Psi}_i\right)^{-1} \mathbf{r}_i } 
\end{align*}
Thus,
\begin{align*}
\frac{\partial f}{\partial s_i}\left(\mathbf{s},y\right) & = \frac{y\mathbf{r}_i^T \left(\mathbf{I}+y\mathbf{\Psi}_i\right)^{-2}\mathbf{r}_i}{\left(1+ys_i\mathbf{r}_i^T\left(\mathbf{I}+y\mathbf{\Psi}_i\right)^{-1} \mathbf{r}_i \right)^2} \\
& = y \mathbf{r}_i^T \left(\mathbf{I}+y\mathbf{R}_s\right)^{-2}\mathbf{r}_i \\
& = y \left[\mathbf{R}^{\frac{1}{2}}\left(\mathbf{I}+y\mathbf{R}_s\right)^{-2}\mathbf{R}^{\frac{1}{2}}\right]_{i,i}
\end{align*}
Finally, 
\begin{align*}
\delta'_i = \frac{\partial \delta}{\partial s_i}&  = - \frac{\delta \left[\mathbf{R}^{\frac{1}{2}}\left(\mathbf{I}+\delta\mathbf{R}_s\right)^{-2}\mathbf{R}^{\frac{1}{2}}\right]_{i,i}}{ \tr\left[\mathbf{R}_{\mathbf{s}}\left(\mathbf{I}+\delta \mathbf{R}_{\mathbf{s}}\right)^{-2}\right]}.
\end{align*}
and
$$
 	\frac{\partial \overline{\text{MSE}}\left(\mathbf{s}\right)}{\partial s_i}=-
 	\frac{\delta  \left[\mathbf{R}^{\frac{1}{2}}\left(\mathbf{I}+\delta\mathbf{R}_s\right)^{-2}\mathbf{R}^{\frac{1}{2}}\right]_{i,i}}{ \tr\left[\mathbf{R}_{\mathbf{s}}\left(\mathbf{I}+\delta \mathbf{R}_{\mathbf{s}}\right)^{-2}\right]}. 
$$

\subsection*{Gradient of the LCE}
The LCE converges a.s. to the following quantity
\begin{align*}
\overline{\text{LCE}}\left(\mathbf{s}\right) & = -\frac{1}{m} \log \det \left(\mathbf{I}+\delta\mathbf{R}_{\mathbf{s}}\right)+\log\left(c\delta\right)+1 \\
& = -\frac{1}{m} \log \left|\mathbf{I}+\delta \mathbf{\Psi}_i\right|
\\
&-\frac{1}{m}\log \left(1+\delta s_i\mathbf{r}_i^T\left(\mathbf{I}+\delta \mathbf{\Psi}_i\right)^{-1}\mathbf{r}_i\right) \\
&   + \log \left(c\delta\right) +1.
\end{align*}
We have
\begin{align*}
\frac{\partial }{\partial s_i} \frac{1}{m} \log \left|\mathbf{I}+\delta \mathbf{\Psi}_i\right| = \frac{\delta_i'}{ m}\tr\left[\mathbf{\Psi}_i\left(\mathbf{I}+\delta \mathbf{\Psi}_i\right)^{-1}\right].
\end{align*}
and
\begin{align*}
&\frac{\partial }{\partial s_i} \frac{1}{m}\log \left(1+\delta s_i\mathbf{r}_i^T\left(\mathbf{I}+\delta\mathbf{\Psi}_i\right)^{-1}\mathbf{r}_i\right)  \\
& = 
\frac{\frac{\delta}{m} \mathbf{r}_i^T\left(\mathbf{I}+\delta \mathbf{\Psi}_i\right)^{-1}\mathbf{r}_i+\frac{\delta'_i s_i}{m} \mathbf{r}_i^T\left(\mathbf{I}+\delta\mathbf{\Psi}_i\right)^{-2}\mathbf{r}_i}{1+\delta s_i\mathbf{r}_i^T\left(\mathbf{I}+\delta \mathbf{\Psi}_i\right)^{-1}\mathbf{r}_i}
\\
& = \frac{\delta}{m}	\left[\mathbf{R}^{\frac{1}{2}}\left(\mathbf{I}+\delta\mathbf{R}_s\right)^{-1}\mathbf{R}^{\frac{1}{2}}\right]_{i,i} \\
& + \frac{\delta_i's_i}{m} \frac{\left[\mathbf{R}^{\frac{1}{2}}\left(\mathbf{I}+\delta\mathbf{R}_s\right)^{-2}\mathbf{R}^{\frac{1}{2}}\right]_{i,i}}{1-\delta s_i\left[\mathbf{R}^{\frac{1}{2}}\left(\mathbf{I}+\delta \mathbf{R}_s\right)^{-1}\mathbf{R}^{\frac{1}{2}}\right]_{i,i}}
\end{align*}
Finally,
\begin{align*}
\frac{\partial \overline{\text{LCE}}\left(\mathbf{s}\right)}{\partial s_i} & = -\frac{\delta_i'}{m}\tr\left[\mathbf{\Psi}_i\left(\mathbf{I}+\delta\mathbf{\Psi}_i\right)^{-1}\right]  \\
& - \frac{\frac{\delta}{m} \mathbf{r}_i^T\left(\mathbf{I}+\delta\mathbf{\Psi}_i\right)^{-1}\mathbf{r}_i+\frac{\delta'_i s_i}{m} \mathbf{r}_i^T\left(\mathbf{I}+\delta\mathbf{\Psi}_i\right)^{-2}\mathbf{r}_i}{1+\delta s_i\mathbf{r}_i^T\left(\mathbf{I}+\delta \mathbf{\Psi}_i\right)^{-1}\mathbf{r}_i} \\
& + \frac{\delta_i'}{\delta}. \\
& = -\frac{\delta_i'}{m}\tr\left[\mathbf{\Psi}_i\left(\mathbf{I}+\delta \mathbf{\Psi}_i\right)^{-1}\right] \\
& - \frac{\delta}{m}	\left[\mathbf{R}^{\frac{1}{2}}\left(\mathbf{I}+\delta\mathbf{R}_s\right)^{-1}\mathbf{R}^{\frac{1}{2}}\right]_{i,i} \\
& - \frac{\delta_i's_i}{m} \frac{\left[\mathbf{R}^{\frac{1}{2}}\left(\mathbf{I}+\delta\mathbf{R}_s\right)^{-2}\mathbf{R}^{\frac{1}{2}}\right]_{i,i}}{1-\delta s_i\left[\mathbf{R}^{\frac{1}{2}}\left(\mathbf{I}+\delta \mathbf{R}_s\right)^{-1}\mathbf{R}^{\frac{1}{2}}\right]_{i,i}}\\
& +  \frac{\delta_i'}{\delta} \\
& = \left(1-c\right)\frac{\delta_i'}{\delta} + \frac{\delta_i'}{m \delta}\tr\left(\mathbf{I}+\delta \mathbf{R}_s\right)^{-1} \\
& - \frac{\delta}{m}	\left[\mathbf{R}^{\frac{1}{2}}\left(\mathbf{I}+\delta\mathbf{R}_s\right)^{-1}\mathbf{R}^{\frac{1}{2}}\right]_{i,i} \\
\end{align*}
\begin{align*}	
	\frac{\partial \overline{\text{LCE}}\left(\mathbf{s}\right)}{\partial s_i}
& 	= \left(1-c\right)\frac{\delta_i'}{\delta} + \frac{\delta_i'}{m \delta}\tr\left(\mathbf{I}+\delta \mathbf{R}_s\right)^{-1} \\
& - \frac{\delta}{m}	\left[\mathbf{R}^{\frac{1}{2}}\left(\mathbf{I}+\delta\mathbf{R}_s\right)^{-1}\mathbf{R}^{\frac{1}{2}}\right]_{i,i} 
\end{align*}
\subsection*{Gradient of the WEV}
The WEV converges a.s. to the following quantity
\begin{align*}
\overline{\text{WEV}}\left(\mathbf{s}\right) = -\frac{1}{\eta} + \frac{1}{m} \tr\left[\mathbf{R}_{\mathbf{s}}\left(\mathbf{I}+\eta \mathbf{R}_{\mathbf{s}}\right)^{-1}\right],
\end{align*}
where 
\begin{align*}
\eta^2 \tr\left[\mathbf{R}_{\mathbf{s}}^2\left(\mathbf{I}+\eta \mathbf{R}_{\mathbf{s}}\right)^{-2}\right]-m=0.
\end{align*}
Define 
\begin{align*}
f\left(\mathbf{s},y\right) & = y^2 \tr\left[\mathbf{R}_{\mathbf{s}}^2\left(\mathbf{I}+y \mathbf{R}_{\mathbf{s}}\right)^{-2}\right]-m \\
& = n-m + \tr \left[\left(\mathbf{I}+y \mathbf{R}_{\mathbf{s}}\right)^{-2}\right] - 2\tr \left[ \left(\mathbf{I}+y \mathbf{R}_{\mathbf{s}}\right)^{-1}\right]
\end{align*}
Then,
\begin{align*}
\frac{\partial f}{\partial y}\left(\mathbf{s},y\right) =2y \tr\left[\mathbf{R}^2_{\mathbf{s}}\left(\mathbf{I}+y \mathbf{R}_{\mathbf{s}}\right)^{-2}\right] - 2y^2 \tr \left[\mathbf{R}^3_{\mathbf{s}}\left(\mathbf{I}+y \mathbf{R}_{\mathbf{s}}\right)^{-3}\right].
\end{align*}
Note that 
\begin{align*}
\left(\mathbf{I}+y \mathbf{R}_{\mathbf{s}}\right)^{-1} & = \left(\mathbf{I}+y\mathbf{\Psi}_i+ys_i \mathbf{r}_i\mathbf{r}_i^T\right)^{-1} \\
& = \left(\mathbf{I}+y\mathbf{\Psi}_i\right)^{-1} \\
& - \frac{y s_i \left(\mathbf{I}+y\mathbf{\Psi}_i\right)^{-1} \mathbf{r}_i\mathbf{r}_i^T \left(\mathbf{I}+y\mathbf{\Psi}_i\right)^{-1}}{1+y s_i \mathbf{r}_i^T \left(\mathbf{I}+y\mathbf{\Psi}_i\right)^{-1} \mathbf{r}_i}
\end{align*}
Thus,
\begin{align*}
\tr	\left(\mathbf{I}+y \mathbf{R}_{\mathbf{s}}\right)^{-2} & = \tr\left(\mathbf{I}+y\mathbf{\Psi}_i\right)^{-2} \\
&+\frac{y^2 s_i^2 \mathbf{r}_i^T\left(\mathbf{I}+y\mathbf{\Psi}_i\right)^{-2} \mathbf{r}_i\mathbf{r}_i^T\left(\mathbf{I}+y\mathbf{\Psi}_i\right)^{-2} \mathbf{r}_i}{\left(1+y s_i \mathbf{r}_i^T \left(\mathbf{I}+y\mathbf{\Psi}_i\right)^{-1} \mathbf{r}_i\right)^2} \\
& -2 \frac{y s_i \mathbf{r}_i^T \left(\mathbf{I}+y\mathbf{\Psi}_i\right)^{-3}\mathbf{r}_i}{1+y s_i \mathbf{r}_i^T \left(\mathbf{I}+y\mathbf{\Psi}_i\right)^{-1} \mathbf{r}_i}.
\end{align*}
Therefore,
\begin{align*}
\frac{\partial f}{\partial s_i}\left(\mathbf{s},y\right) &  = \frac{2y \mathbf{r}_i^T \left[\left(\mathbf{I}+y\mathbf{\Psi}_i\right)^{-2}-\left(\mathbf{I}+y\mathbf{\Psi}_i\right)^{-3}\right] \mathbf{r}_i}{\left(1+y s_i \mathbf{r}_i^T \left(\mathbf{I}+y\mathbf{\Psi}_i\right)^{-1} \mathbf{r}_i\right)^2} \\
& + \frac{2y^2 s_i \left[\mathbf{r}_i^T \left(\mathbf{I}+y\mathbf{\Psi}_i\right)^{-2}\mathbf{r}_i\right]^2}{\left(1+y s_i \mathbf{r}_i^T \left(\mathbf{I}+y\mathbf{\Psi}_i\right)^{-1} \mathbf{r}_i\right)^3}.
\end{align*}
Thus, 
\begin{align*}
\eta_i' & = \frac{\partial \eta}{\partial s_i} \\
& = -  \left(2\eta \tr\left[\mathbf{R}^2_{\mathbf{s}}\left(\mathbf{I}+\eta \mathbf{R}_{\mathbf{s}}\right)^{-2}\right] - 2\eta^2 \tr\left[\mathbf{R}^3_{\mathbf{s}}\left(\mathbf{I}+\eta \mathbf{R}_{\mathbf{s}}\right)^{-3}\right]\right)^{-1} \\
& \times \Biggl[\frac{2\eta \mathbf{r}_i^T \left[\left(\mathbf{I}+\eta\mathbf{\Psi}_i\right)^{-2}-\left(\mathbf{I}+\eta\mathbf{\Psi}_i\right)^{-3}\right] \mathbf{r}_i}{\left(1+\eta s_i \mathbf{r}_i^T \left(\mathbf{I}+\eta\mathbf{\Psi}_i\right)^{-1} \mathbf{r}_i\right)^2} \\
& + \frac{2\eta^2 s_i \left[\mathbf{r}_i^T \left(\mathbf{I}+\eta\mathbf{\Psi}_i\right)^{-2}\mathbf{r}_i\right]^2}{\left(1+\eta s_i \mathbf{r}_i^T \left(\mathbf{I}+\eta\mathbf{\Psi}_i\right)^{-1} \mathbf{r}_i\right)^3}\Biggr] \\
& =  -  \left( \tr\left[\mathbf{R}_{\mathbf{s}}\left(\mathbf{I}+\eta \mathbf{R}_{\mathbf{s}}\right)^{-2}\right] - \tr\left[\mathbf{R}_{\mathbf{s}}\left(\mathbf{I}+\eta \mathbf{R}_{\mathbf{s}}\right)^{-3}\right]\right)^{-1} \\
& \times \eta\left( \left[\mathbf{R}^{\frac{1}{2}}\left(\mathbf{I}+\eta\mathbf{R}_s\right)^{-2}\mathbf{R}^{\frac{1}{2}}\right]_{i,i}- \left[\mathbf{R}^{\frac{1}{2}}\left(\mathbf{I}+\eta\mathbf{R}_s\right)^{-3}\mathbf{R}^{\frac{1}{2}}\right]_{i,i}\right) 
\end{align*}
Finally,
\begin{align*}
\frac{\partial \overline{\text{WEV}}\left(\mathbf{s}\right)}{\partial s_i} & =-\frac{\left(c-1\right)\eta_i'}{\eta^2}+\frac{\eta_i'}{m\eta^2} \\
& \times \left[2\tr\left(\mathbf{I}+\eta\mathbf{R}_s\right)^{-1}-\tr\left(\mathbf{I}+\eta\mathbf{R}_s\right)^{-2}\right] \\
&+\frac{1}{m}\left[\mathbf{R}^{\frac{1}{2}}\left(\mathbf{I}+\eta\mathbf{R}_s\right)^{-2}\mathbf{R}^{\frac{1}{2}}\right]_{i,i}
\end{align*}
\section*{Appendix B}
%
%
\section*{Proof of Theorem 1}
We prove the convexity of the $\overline{\text{MSE}}$ using a similar argument as the one proposed in \cite{dumont}. 
Let $p$ be a non negative integer, and $\mathbf{\widetilde{W}}$ be a $np\times mp$ matrix with i.i.d $\mathcal{CN}\left(0,1\right)-$ distributed entries. For $\mathbf{s} \in \mathbb{R}^n_+$, let $\mathbf{R}\left(\mathbf{s}\right) = \mathbf{R}^{\frac{1}{2}} \diag\left(\mathbf{s}\right)\mathbf{R}^{\frac{1}{2}}$ and
\begin{align*}
 \widetilde{\mathbf{R}}\left(\mathbf{s}\right)=\mathbf{I}_p \otimes \mathbf{R}\left(\mathbf{s}\right),
\end{align*}
where $\otimes$ is the Kronecker product of matrices. Then, $\widetilde{\mathbf{R}}$ is a $np \times np$ positive definite block diagonal matrix.  Since $p$ affects both dimensions $m$ and $n$ at the same pace, we have the following convergence
\begin{equation}
\lim_{p \to \infty} \tr \left[\left(\mathbf{\widetilde{W}}^H\widetilde{\mathbf{R}}\left(\mathbf{s}\right) \mathbf{\widetilde{W}}\right)^{-1}\right] = \widetilde{\text{MSE}}\left( \mathbf{s}\right),
\end{equation}
where $ \widetilde{\text{MSE}}\left( \widetilde{\mathbf{R}}\left(\mathbf{s}\right)\right) = \widetilde{\delta}\left(\mathbf{s}\right)$ is solution to the following fixed-point equation
\begin{equation}
	\widetilde{\delta}\left(\mathbf{s}\right) = \frac{mp}{\tr \left[\widetilde{\mathbf{R}}\left(\mathbf{s}\right)\left(\mathbf{I}+	\widetilde{\delta}\left(\mathbf{s}\right) \widetilde{\mathbf{R}}\left(\mathbf{s}\right)\right)^{-1}\right]}.
\end{equation}
Using the block-diagonal structure of $\widetilde{\mathbf{R}}\left(\mathbf{s}\right)$, we have
\begin{align*}
	\left(\mathbf{I}+	\widetilde{\delta}\left(\mathbf{s}\right) \widetilde{\mathbf{R}}\left(\mathbf{s}\right)\right)^{-1} = \mathbf{I}_p \otimes \left(\mathbf{I}+\widetilde{\delta}\left(\mathbf{s}\right) \mathbf{R}\left(\mathbf{s}\right)\right)^{-1}.
\end{align*}
Then,
\begin{align*}
	\widetilde{\mathbf{R}}\left(\mathbf{s}\right)\left(\mathbf{I}+	\widetilde{\delta}\left(\mathbf{s}\right) \widetilde{\mathbf{R}}\left(\mathbf{s}\right)\right)^{-1} & = \left[\mathbf{I}_p \otimes \mathbf{R}\left(\mathbf{s}\right)\right]\\
	& \times \left[ \mathbf{I}_p \otimes \left(\mathbf{I}+\widetilde{\delta}\left(\mathbf{s}\right) \mathbf{R}\left(\mathbf{s}\right)\right)^{-1}\right] \\
	& = \mathbf{I}_p \otimes \left[\mathbf{R}\left(\mathbf{s}\right)\left(\mathbf{I}+\widetilde{\delta}\left(\mathbf{s}\right) \mathbf{R}\left(\mathbf{s}\right)\right)^{-1}\right]
\end{align*}
Thus, 
\begin{align*}
	\tr \left[\widetilde{\mathbf{R}}\left(\mathbf{s}\right)\left(\mathbf{I}+	\widetilde{\delta}\left(\mathbf{s}\right) \widetilde{\mathbf{R}}\left(\mathbf{s}\right)\right)^{-1}\right] = p \tr \left[\mathbf{R}\left(\mathbf{s}\right)\left(\mathbf{I}+\widetilde{\delta}\left(\mathbf{s}\right) \mathbf{R}\left(\mathbf{s}\right)\right)^{-1}\right]
\end{align*}
and, 
\begin{align*}
		\widetilde{\delta}\left(\mathbf{s}\right) & = \frac{m}{\tr \left[\mathbf{R}\left(\mathbf{s}\right)\left(\mathbf{I}+\widetilde{\delta}\left(\mathbf{s}\right) \mathbf{R}\left(\mathbf{s}\right)\right)^{-1}\right]} \\
		& = \overline{\text{MSE}}\left(\mathbf{s}\right).
\end{align*}
In other words,
\begin{equation}
\lim_{p \to \infty} \tr \left[\left(\mathbf{\widetilde{W}}^H\widetilde{\mathbf{R}}\left(\mathbf{s}\right) \mathbf{\widetilde{W}}\right)^{-1}\right] = \overline{\text{MSE}}\left(\mathbf{s}\right).
\end{equation}
As a matter of fact, $\overline{\text{MSE}}\left(\mathbf{s}\right)$ is the pointwise limit of convex functions, which implies that it is also convex. The proof of convexity of the LCE and the WEV follow using the same argument. This concludes the proof of Theorem 1.
\\
\end{document}